\documentclass[letterpaper,11pt]{amsart}
\pdfoutput=1
\usepackage{xcolor}
\definecolor{darkblue}{rgb}{0.,0.,0.75}
\definecolor{darkred}{rgb}{0.5,0.,0.}
\definecolor{darkgreen}{rgb}{0.0,0.5,0.}
\usepackage[pdftex, colorlinks=true,linkcolor=darkblue,citecolor=darkred,urlcolor=blue]{hyperref}
\usepackage{amsmath,amsthm,amsfonts,amssymb,braket,comment}
\usepackage[all]{xy}
\usepackage{fullpage}
\usepackage{chemarrow}
\usepackage{marginnote}
\usepackage[numbers]{natbib}
\usepackage{tikz-cd}
\usepackage{scalerel}
\newtheorem{lemma}{Lemma}[section]
\newtheorem{theorem}[lemma]{Theorem}

\newtheorem{corollary}[lemma]{Corollary}

\newtheorem{proposition}[lemma]{Proposition}
\newtheorem{proposition-definition}[lemma]{Proposition-Definition}
\theoremstyle{definition}

\usepackage{mathtools}
\def\CF {{\mathcal F}}

\def\CL {{\mathcal L}}

\def\h{{\mbox{H}}}
\def\L{{\mbox{L}}}
\def\hm{{\mbox{H}}^{-}}
\def\E{{\mbox{E}}}
\def\ifpd{\IF_p^{\dd}}

\def\IF{\mathbb{F}}

\def\IC{\mathbb{C}}

\def\IR{{\mathbb{R}}}

\def\IZ{{\mathbb{Z}}}

\def\la{\langle}
\def\ra{\rangle}

\def\gl{{\mbox{GL}}}
\def\egl{{\mbox{EGL}}}
\def\La{{\Lambda}}
\def\l{{\lambda}}
\def\lm{{\lambda}^{-}}
\def\sc{\scaleto{\,\circ\,}{4pt}}
\def\sl{\mathcal{S}_{\textsuperscript{L}}}
\def\sls{\mathcal{S}_{\textsuperscript{L}^*}}

\def\LL{\CL_{\scaleto{\L}{5pt}}}
\def\FL{\CF_{\scaleto{\L}{5pt}}}
\def\OLL{\Omega\CL_{\scaleto{\L}{5pt}}}
\def\mas{\mbox{Mas}}
\def\st{\mbox{st}}
\def\Det{\mbox{Det}}
\def\P{\mbox{P}}
\def\U{\mbox{U}}
\def\K{\mbox{K}}
\newcommand{\hp}{{\mathsf{U}}}
\newcommand{\ehp}{{\mathsf{EU}}}

\newcommand{\bass}{{\mathsf{B}}}

\newcommand{\switt}{{\mathfrak{W}}}
\newcommand{\umod}{{\mathfrak{U}}}

\newcommand{\vtheory}{{\mathfrak{L}}}

\newcommand{\dd}{{\mathsf{d}}} % space dimension

\begin{document}

\title[Homotopy Classification of Clifford]
{Homotopy Classification of loops of Clifford unitaries}
\author{Roman Geiko}
\address{NHETC and Department of Physics and Astronomy, Rutgers University}
\author{Yichen Hu}
\address{Department of Physics, Princeton University}

\begin{abstract}
Clifford quantum circuits are elementary invertible transformations of quantum systems that map Pauli operators to Pauli operators. We study periodic one-parameter families of  Clifford circuits, called loops of Clifford circuits, acting on $\dd$-dimensional lattices of prime $p$-dimensional qudits. We propose to use the notion of algebraic homotopy to identify topologically equivalent loops. We calculate homotopy classes of such loops for any odd $p$ and $\dd=0,1,2,3$, and $4$. Our main tool is the Hermitian K-theory, particularly a generalization of the Maslov index from symplectic geometry. We observe that the homotopy classes of loops of Clifford circuits in $(\dd+1)$-dimensions coincide with the quotient of the group of Clifford Quantum Cellular Automata modulo shallow circuits and lattice translations in $\dd$-dimensions.
\end{abstract}

\maketitle

\section{Introduction}\label{intro}

Unitary dynamics of many-body quantum systems constitute one of the core components of quantum dynamics in the theory of quantum computation and information processing. Recent works have made tremendous progress in understanding the entanglement generation \cite{Nahum2017, Nahum2018,Li2019, Fisher2023}, scrambling of information \cite{Mi2021}, and localization \cite{farshi2022mixing,farshi2023absence} by generic (random) circuits. Unitary dynamics, more general than finite layers of quantum circuits, can be found in unitaries preserving strict locality, also known as quantum cellular automata (QCA) \cite{schlingemann2008structure, Freedman2020}: these unitaries map local observables to local observables. QCA are intrinsically related to topological phases of matter and play an important role in the construction and classification thereof \cite{Hastings2013, Lukasz2020}. In fact, one can study topological phases of unitary dynamics on their own -- we say that two QCA belong to the same topological phase of unitary dynamics if they differ by a quantum circuit and/or lattice translation. Phases of Clifford unitary dynamics have been completely analyzed in \cite{Haah2022}. The main subject of our paper is the analysis of phases of periodic one-parameter families of unitary dynamics, loops of Clifford circuits, sometimes going by the name of Clifford Floquet circuits.

The notion of a topological phase of matter colloquially means a family of systems whose elements can be connected by adiabatic evolution. A mathematically rigorous definition of a topological phase involves sophisticated functional analysis and is an area of active research, see \cite{ogata2021classification, kapustin2022local,beaudry2023homotopical,chung2023topological} and references therein.  
In this paper, we use algebraic tools  to study homotopy groups (in a sense explained below) of loops of Clifford unitary dynamics bypassing functional analysis. We assume that the unitary dynamics is strictly locality-preserving, Clifford, and translation invariant. By strict locality-preserving dynamics, we mean operator algebra automorphisms that map operators with finite support to operators with finite support. Clifford dynamics stands for automorphisms of the operator algebra normalizing the group of Pauli operators. A natural way to study Clifford unitary dynamics is by analyzing its action on stabilizer Hamiltonians. As was demonstrated by J. Haah, stabilizer Hamiltonians admit a description in terms of certain submodules of the so-called Pauli module, a module over the ring of Laurent polynomials \cite{haah2013}, while Clifford QCA correspond to automorphisms of the Pauli module. We shall be interested in specifically simple Lagrangian submodules of the Pauli module. In the language of stabilizer Hamiltonians, the set of Lagrangian submodules, the Lagrangian Grassmannian, corresponds to a subset of invertible Hamiltonians, i.e., Hamiltonians having a unique ground state on a lattice of any topology.

We argue that the phases of Clifford unitary dynamics can be analyzed by studying the homotopy theory of the Lagrangian Grassmannian. In particular, we analyze topological phases of loops of Clifford circuits by studying periodic paths in the Lagrangian Grassmanian. The group structure on the set of loops of Clifford circuits is induced by stacking corresponding families of invertible Hamiltonians. This group law differs from the obvious group structure on the set of topological phases of unitary dynamics, which is induced by the composition of QCA. We analyze periodic paths in the Lagrangian Grassmannian using a generalization of the classical Maslov index, which is a homotopy invariant of loops in the real Lagrangian Grassmannian. The version of the Maslov index used in this paper is a manifestation of the fundamental theorem of Hermitian K-theory \cite{Karoubi1980} of Karoubi in the interpretation of Barges and Lannes \cite{BL}. Since the fundamental theorem of Hermitian K-theory is applicable only to the Laurent polynomial rings with coefficients in $\IF_p$ with $p>2$, our results  do not apply to lattices of qubits.

The main result of this paper is a classification of homotopy classes of loops of Clifford circuits for translationally invariant $\dd$-dimensional lattices of prime $p$-dimensional qudits with $\dd=0,1,2,3$, and $4$, and with $p$  greater than $2$, which we denote by $\Omega \mathfrak C(\dd,p)$:
\begin{align*}
\Omega \mathfrak C(0,p)\cong \Omega \mathfrak C(1,p) \cong \Omega \mathfrak C(2,p) \cong \Omega \mathfrak C(3,p) \cong 0\,,\\
    \Omega \mathfrak C(4,p) \cong \begin{cases}
        \IZ/2\IZ\oplus \IZ/2\IZ\,,\quad p\equiv 1 \;\mbox{mod}\; 4\,,\\
        \IZ/4\IZ\,,\quad \quad \quad \;\;\quad p\equiv 3 \;\mbox{mod}\; 4\,.\\
    \end{cases}
\end{align*}

This paper is organized as follows. In Section \ref{section:review} we give an extended introduction to the symplectic formalism for Clifford QCA, starting with zero-dimensional Clifford unitaries. For $\dd>0$, we introduce a Pauli algebra and Pauli modules such that Clifford QCA are given by automorphisms of the Pauli modules. In section \ref{sec:algebraicintro} we set the notation for modules equipped with (anti-) hermitian forms and module automorphisms, as well as review some  elementary facts about algebraic and Hermitian K-theory. In Section \ref{sec:pathsoflagrnagians} we introduce the notion of algebraic homotopies of Lagrangian submodules, discuss its basic properties and relate it to the elementary automorphisms of the Pauli module. In Section \ref{interlude} we pause for a discussion of how the algebraic homotopy equivalence between Lagrangian stabilizer modules corresponds to the CFDQC equivalence between the stabilizer Hamiltonians . Then, we review the classical Maslov index for loops of Lagrangians in a real phase space. In Section \ref{sec:loopsoflagrangians} we continue the analysis of loops of Lagrangians inside Pauli modules using a generalization of the Maslov index. In Section \ref{sec:fundideal} we review some known results about the Witt group of hermitian forms and compute relevant fundamental ideals. We make concluding remarks in Section \ref{sec:discussion}.

We use some classical results from Hermitian K-theory, like its homotopy invariance, without proofs. We provide proofs for less standard facts that constitute the computation of the homotopy classes of loops of Lagrangians. Many important results have been reviewed in \cite{Haah2022}, where we inherited our system of notations from.

{\bf Acknowledgments.}
We thank P. Balmer, J. Fasel, P. Orson, and W. Pitsch for correspondence and M. Levin and Y. Liu for discussions. We thank J. Haah and T. Mainiero for comments on the draft. R.G. is grateful to C. Weibel for a  hint and to G. Shuklin for many valuable discussions. R.G. also thanks the organizers of the summer school ``Mathematics of Topological Phases of Matter" where he learned about QCA. The work of R.G. was supported by the US Department of Energy under grant DE-SC0010008. Y.H. would like to thank support from 
NSF through the Princeton University’s Materials Research Science and Engineering Center DMR-2011750 and the Gordon and Betty Moore Foundation through Grant
GBMF8685 towards the Princeton theory program.

 \section{Review of symplectic formalism}\label{section:review}
Throughout this letter, we assume that the lattice systems under consideration are translation-invariant and unitaries acting on lattice systems are translation invariant as well.  We focus on Hamiltonians made of products of Pauli operators and unitaries mapping products of Pauli operators in products of Pauli operators, also known as Clifford unitaries. In this section we review the formalism of symplectic vector spaces and symplectic modules as a convenient framework for translation-invariant Pauli Hamiltonians and Clifford unitaries.

{\bf Single qudit.} We begin our review with the case of a single site. The Pauli operators acting on a prime $p$-dimensional qudit form the Pauli group, which we denote by $\P_p$. This group is defined through generators and relations: 
\begin{align}\label{Def:CliffGroup}
   \P_p\coloneqq \la X, Z ,\omega \;|\; X^p=Z^p=\omega^p=1,\; XZ=\omega ZX\ra \,.
\end{align}
Let us pick an orthonormal basis $\{|n\ra\}_{n=0}^{p-1}$ for $\IC^{p}$, then the minimal faithful unitary representation of $\P_p$, i.e., a homomorphism into the unitary group of $\IC^{p}$, denoted by $\U_p$, is given by
\begin{align*}
    \omega |n\ra =e^{2\pi i/p} |n\ra \,,\quad X|n\ra =|n+\scaleto{1}{6pt} \,\mbox{mod}\;p\ra\,,\quad Z |n\ra =e^{2\pi i n/p}|n\ra\,.
\end{align*}
In this paper we are concerned with automorphisms of algebras of operators acting on lattice systems. For the case of a single qudit, all automorphisms of the algebra of observables are given by conjugations with some unitary operator. The group of Clifford unitaries $\mbox{CU}_p$ is a subgroup of the unitary group that normalizes $\mbox{P}_p$ (more accurately, the image of $\mbox{P}_p$ in $\U_p$):
\begin{align*} 
\mbox{CU}_p\coloneqq\{x\in \mbox{U}_p\,|\,x \,\mbox{P}_p\,x^{-1}=\mbox{P}_p\}\,.
\end{align*}
In other words, Clifford unitaries generate the automorphisms of the qudit operator algebra that preserve the Pauli group. 

 Let us recall that the adjoint action of the unitary group is not faithful -- the kernel of this action is the center of $\U_p$ and it is isomorphic to $\U_1$. In its turn, the adjoint action of the quotient group $\mbox{U}_d/\mbox{U}_1$, also known as the projective unitary group $\mbox{PU}_d$, on the matrix algebra is faithful. In order to get rid of this redundancy, we define the projective Clifford group $\mbox{PCU}_p$:
\begin{align*} 
\mbox{PCU}_p\coloneqq\{x\in \mbox{PU}_p\,|\,x \,\mbox{P}_p\,x^{-1}=\mbox{P}_p\}\,.
\end{align*}
whose adjoint action on $\P_p$ is faithful. There is an alternative description
of $\mbox{PCU}_d$, which will be important for future applications, as the normalizer of the abelianization of $\mbox{P}_d$ within $\mbox{PU}_p$. Let $\mbox{AP}_p\coloneqq \P_p/[\mbox{P}_p,\mbox{P}_p]\cong \IF_p\oplus\IF_p$\footnote{Throughout the paper $\IF_p$ stands for a prime field, though we do not use the field structure and we could equivalently use the ring $\IZ/p\IZ$.} be the abelianization of $\P_p$ which is obtained by projecting out the powers of $\omega$ such that the elements of $\mbox{PCU}_p$ are given by powers of $\tilde X$ and $\tilde Z$ such that $\tilde X^p=1$, $\tilde Z^p=1$, and $\tilde X \tilde Z=\tilde Z \tilde X$. Thus, the alternative definition of  $\mbox{PCU}_p$ is 
\begin{align*} 
\mbox{PCU}_p=\{x\in \mbox{PU}_p\,|\,x \,\mbox{AP}_p\,x^{-1}=\mbox{AP}_p\}\,.
\end{align*}
Later on, we shall be analyzing the action of Clifford unitaries on the abelian group $\mbox{AP}_p$. Clearly, $\mbox{AP}_p$ is a subset of $\mbox{PCU}_p$ and its adjoint action on itself is trivial. Therefore, we define the {\bf quotient Clifford group} whose adjoint action on $\mbox{AP}_p$ is faithful:
\begin{align}
    \mbox{C}_p\coloneqq \mbox{PCU}_p/\mbox{AP}_p\,.
\end{align}
A crucial fact about the quotient Clifford group is that it is isomorphic to the group $\mbox{Sp}(2;\IF_p)$ of invertible transformations of the two-dimensional $\IF_p$-vector space which preserve the standard symplectic form $\l=\scaleto{\begin{pmatrix}
    0 & 1\\
    -1 & 0
\end{pmatrix}}{16pt}$ \cite{kitaev2002classical, Havlicek2002}. Any element of $\mbox{Sp}(2;\IF_p)$ corresponds to an automorphism of the operator algebra acting on a qudit that maps a Pauli operator to a Pauli operator with no fixed points. It is instructive to describe the action of $\mbox{C}_p$ on $\mbox{AP}_p$ explicitly: it is simply the regular representation of $\mbox{Sp}(2;\IF_p)$ on $\IF_p\oplus \IF_p$:
\begin{align}
        \mbox{C}_p: \mbox{AP}_p\to \mbox{AP}_p\,,\quad  \begin{pmatrix}
            x\\
            z 
        \end{pmatrix}\mapsto \Phi \cdot \begin{pmatrix}
            x\\
            z 
        \end{pmatrix}
\end{align}
where $\begin{pmatrix}
    x\\ z
\end{pmatrix}\in  \IF_p\oplus \IF_p \cong \mbox{AP}_p$ and $\Phi\in \mbox{Sp}(2;\IF_p) \cong \mbox{C}_p$. 

This result admits a simple interpretation. Let $y_1$ and $y_2$ be a pair of elements from $\mbox{AP}_p$, which correspond to Pauli operators up to phases, then the commutator of the corresponding operators within $\mbox{P}_p$ is given by $\omega^{\l(y_1,y_2)}$, where $\omega$ is the $p$-th root of unity. Thus, we trade $\mbox{P}_p$ for the abelian group $\mbox{AP}_p$ equipped with the symplectic form, while the action of the quotient Clifford group is an automorphism of $\mbox{AP}_p$ preserving this symplectic form. 

\textbf{$N$ qudits.} We treat $N$ copies of $p$-dimensional qudits in the complete analogy with one qudit: the Pauli group $\P_p^N$ consists of $N$ types of species of $X$ and $Z$ such that different species commute. The abelianization of the Pauli group in this case  $\mbox{AP}_p^N\cong \IF_p^N\oplus \IF_p^N$, is a $2N$-dimensional $\IF_p$-vector space, while the quotient Clifford group $\mbox{C}_p^N$ is isomorphic to $\mbox{Sp}(2N;\IF_p)$, the group of invertible transformations preserving the standard symplectic form  $\l_N=\scaleto{\begin{pmatrix}
    0 & 1_N\\
    -1_N & 0
\end{pmatrix}}{16pt}$. 

\textbf{$\dd>0$ dimensions.} Luckily, the symplectic formalism can be extended to translationally invariant many-body systems in any number of spatial dimensions. Let us consider a cubic lattice $\IZ^{\dd}$ with $N$ copies of $p$-dimensional qubits attached to each lattice site. In other words, each site $x\in \IZ^\dd$  supports the Pauli group of $N$ qudits $\P_p^N$. We notice that $\IZ^{\dd}$ acts on itself by translations. Moreover, the group ring $\IF_p(\IZ^{\dd})$ acts on the set of Pauli operators. For example, let us consider a one dimensional lattice with a single qudit per site. Invertible elements of $\IF_p(\IZ)$ are given by monomials of the form $a\,t$ with some $a\in \IF_p^{\times}$ and $t\in \IZ$; let $X_s$ be a Pauli operator supported at site $s\in \IZ$, then $a\,t$ maps it to the Pauli operator $(X_{s+t})^a$. Note that in order to preserve the commutation relations, $a\,t$ must map $Z_s$ to $Z^{a^{-1}}_{s+t}$. The case of $a=1$ corresponds to pure translations, while transformations with $a\neq 1$ we shall call generalized translations. 

Let us notice that the group ring $\IF_p(\IZ^\dd)$ is isomorphic to the ring of Laurent polynomials with coefficients in $\IF_p$:
\begin{align}
\IF_p(\IZ^\dd)\cong \IF_p^{\dd}\coloneqq \IF_p[x_1,x_1^{-1},\dots, x_\dd,x_{\scaleto{\dd}{5pt}}^{-1}]\coloneqq \{ \sum_{i_1,\dots, i_\dd}x_1^{i_1}\dots x_{{\scaleto{\dd}{5pt}}}^{i_{\scaleto{\dd}{3.5pt}}}a_{i_1,\dots, i_{{\scaleto{\dd}{5pt}}}} \;|\;a_{i_1,\dots, i_{{\scaleto{\dd}{5pt}}}}\in \IF_p\}\; .  
\end{align}
Later on we shall be using the polynomial notation. The ring $\IF_p^{\dd}$ is called the {\bf Pauli algebra}. Further, we notice that up to phases, the Pauli operators for $N$ qudits  form a group isomorphic to $(\ifpd)^{2N}$. Clearly, the Pauli algebra $\ifpd$ acts on $(\ifpd)^{2N}$ and we call the latter the {\bf Pauli module}. This is a straightforward generalization of our previous observation that abelianization of the Pauli group on a single site is isomorphic to $(\IF_p)^{2N}$. 

The Pauli algebra, as a group ring, has an innate involution corresponding to the reflection around the origin in $\IZ^\dd$. In the language of Laurent polynomials, the involution maps $x_1\to x_1^{-1},\dots , x_\dd\to x_\dd^{-1}$. The involution allows us to define an $\ifpd$-valued anti-hermitian form on the Pauli module:
\begin{multline}\label{def:pairing}
  \la  \sum_{i_1,\dots, i_\dd}x_1^{i_1}\dots x_{{\scaleto{\dd}{5pt}}}^{i_{\scaleto{\dd}{3.5pt}}}y_{i_1,\dots, i_{{\scaleto{\dd}{5pt}}}},  \sum_{j_1,\dots, j_\dd}x_1^{j_1}\dots x_{{\scaleto{\dd}{5pt}}}^{j_{\scaleto{\dd}{3.5pt}}}y'_{j_1,\dots, j_{{\scaleto{\dd}{5pt}}}}\ra =\\  \sum_{i_1,\dots, i_\dd}\sum_{j_1,\dots, j_\dd}x_1^{-i_1}\dots x_{{\scaleto{\dd}{5pt}}}^{-i_{\scaleto{\dd}{3.5pt}}}x_1^{j_1}\dots x_{{\scaleto{\dd}{5pt}}}^{j_{\scaleto{\dd}{3.5pt}}}\; \l_N(y_{i_1,\dots, i_{{\scaleto{\dd}{5pt}}}},y'_{j_1,\dots, j_{{\scaleto{\dd}{5pt}}}})\,,
\end{multline}
where $y_{i_1,\dots, i_{{\scaleto{\dd}{5pt}}}}$ and $y'_{j_1,\dots, j_{{\scaleto{\dd}{5pt}}}}$ are elements of $\IF_p^{2N}$.

As a group ring, $\ifpd$ comes with augmentation, a map which picks the coefficient of zero degree:
\begin{align}
    \varepsilon:\ifpd \to \IF_p\,,\quad &\varepsilon\big\{ \sum_{i_1,\dots, i_\dd}x_1^{i_1}\dots x_{{\scaleto{\dd}{5pt}}}^{i_{\scaleto{\dd}{3.5pt}}}a_{i_1,\dots, i_{{\scaleto{\dd}{5pt}}}}\big\} =a_{0,\dots, 0}\,.
\end{align}
To this end, we define an $\IF_p$-valued anti-hermitian form by composing \eqref{def:pairing} and $\varepsilon$. Here and everywhere below we assume that the Pauli module is equipped with the $\ifpd$-valued anti-hermitian form \eqref{def:pairing}. Let, for example, $y_1$ and $y_2$ be a pair of elements of the Pauli module, which correspond to Pauli operators up to phases: the commutator of the corresponding Pauli operators is given by $e^{\frac{2\pi i}{p}\varepsilon\{\la y_1,y_2 \ra\}}$. Therefore, we can replace the set of Pauli operators on the lattice with the Pauli module equipped with the standard anti-hermitian form.

\textbf{Basis-free notation.} Later we shall abstract away the concrete form of the Pauli module disguising it as an $\ifpd$-module of the form $\L\oplus \L^*$ where $\L$ is a $\ifpd$-module free module and $\L^*$ is its dual. The direct summand $\L$ should be seen as a module containing all species of the Pauli $X$ operators, while $\L^*$ contains all the Pauli $Z$ operators. Note that we need the dual module in order to have a pairing between Pauli $X$ and Pauli $Z$ operators, similarly to the conjugate variables in Hamiltonian mechanics. We have seen above that Pauli $X$ and Pauli $Z$ operators transform differently under generalized translations. Precisely, the group of automorphisms of $\L$ acts on $\L^{*}$ as on a dual module.

\textbf{Automorphisms of Pauli modules.} The Pauli module is a free module equipped with the standard anti-hermitian form \eqref{def:pairing}. The automorphism of the Pauli module is an invertible transformation that preserves the anti-hermitian form -- we call the group of automorphisms of the Pauili module the group of $\lm$-unitaries. It turns out the action of the quotient Clifford group on Pauli operators in the many-body case is isomorphic to the action of the group of $\lm$-unitaries on the Pauli module, see \cite{schlingemann2008structure}. Note that we abstain from using the term ``symplectic" in the many-body case, as we keep track of the involution on the Pauli algebra.  

An implicit feature of the formalism of $\lm$-unitaries is that any $\lm$-unitary has a property that it maps operators localized in a finite region to operators localized in a finite region. This property follows from the fact that we are operating with the ring of Laurent \textit{polynomials}: degree of a polynomial roughly corresponds to the radius of action of such a polynomial on a Pauli operator. 

\subsection{Stabilizer modules}This paper is devoted to the analysis of $\lm$-unitaries and it is convenient to study them by their action on submodules of Pauli modules, more precisely, on Lagrangian submodules. Thus, we briefly digress on the so-called stabilzer modules following \cite{haah2013,haah2016}.

Let us consider a translation invariant Hamiltonian 
\begin{equation}\label{Hamiltonian}
    H=-\sum_{t\in \mathbb{Z}^{\dd}}(h_{1,t}+\cdots+h_{v,t})
\end{equation}
where $h_{1,t},\dots, h_{v,t}$ are terms corresponding to the interaction types that satisfy a number of properties: we assume that the terms $h_{1,t},\dots, h_{v,t}$ are made of Pauli operators, mutually commute, and $H$ is frustration-free -- such Hamiltonians are called {\bf stabilizer Hamiltonians}. Since the Hamiltonian in \eqref{Hamiltonian} is obtained by all possible translations of these $v$ terms, we can view these terms as generators for an $\ifpd$-module. The module generated by $h_{1,t},\dots, h_{v,t}$ is called {\bf stabilizer module}.

\paragraph{{\bf Example}}\label{example}As an example, let us consider a one-dimensional spin-chain with a single $p$-dimensional qudit per site. Here, the Pauli algebra is $\IF_p^1=\IF_p[x,x^{-1}]$ and the Pauli module is two-dimensional. Let $H$ be a generalized $\IZ/d\IZ$ cluster Hamiltonian:
\begin{align}\label{cluster}
    H=\sum_{x\in \IZ}X_{x-1}Z_x X_{x+1}+h.c.
\end{align}
The stabilizer module for such a Hamiltonian is generated by $\begin{pmatrix}x^{-1}+x^{+1}\\
1\end{pmatrix}$ and $\begin{pmatrix}-x^{-1}-x^{+1}\\
-1\end{pmatrix}$, which is, in fact, one-dimensional.

The commutation relations between Pauli operators are encoded by the pairing \eqref{def:pairing}. By definition, all elements of the stabilizer module mutually commute what means that the stabilizer module is {\bf isotropic}, i.e., pairing \eqref{def:pairing} degenerates on such a submodule. Further, we can ask -- what Pauli operators commute  with the terms in $H$, i.e., what are the local symmetries admitted by $H$? Let us denote the stabilizer module by $S$, then the operators commuting with $S$ form the module $S^{\perp}$ where orthogonal complement is taken with respect to the pairing \eqref{def:pairing} within the Pauli module. The isotropy condition $S\subset S^{\perp}$ is implied by the definition of the stabilizer module, while the absence of local symmetries is equivalent to the condition $S^{\perp}\subset S$, i.e., $S$ is {\bf coisotropic}. Any Hamiltonian corresponding to a coisotropic stabilizer module is said to satisfy the {\bf local topological order} condition. There is a special class of locally topologically ordered Hamiltonians having a unique ground state on any topology -- we call such Hamiltonians invertible. Even more restricted is the class of Hamiltonians corresponding to Lagrangian stabilizer modules defined in Section \ref{sec:lagrangians}: they are isotropic, coisotropic and direct summands in the Pauli module. One checks that the module from the example above is indeed a Lagrangian. We recommend \cite{haah2013} and \cite{haah2016} for further reading on how properties of stabilizer modules translate into physical properties of Hamiltonians. The upshot of this discussion is that Lagrangian submodules of the Pauli module correspond to invertible Hamiltonians, though not all invertible systems can be encoded by a Lagrangian submodule. 

\subsection{Clifford QCA}  A quantum cellular automaton (QCA) is a locality-preserving automorphism of the observables algebra: it maps operators localized in finite regions to operators localized in finite regions. Clifford QCA (CQCA) are QCA which map product of Pauli operators to products of Pauli operators. According to the discussion above, we identify the groups of Clifford QCA and the group of $\lm$-unitaries.

The most elementary CQCA are given by conjugations with Clifford gates -- local unitary transformations whose adjoint action maps Pauli operators to Pauli operators. An operator evolution generated by mutually commuting Clifford gates is called Clifford quantum circuit (CQC). A finite number of layers of CQCs is called Clifford finite-depth (shallow) quantum circuit (CFDQC), which is a prototypical example of a CQCA. While the group of CQCA correspond to the group of $\lm$-unitaries, CFDQCs correspond to the normal subgroup of elementary $\lm$-unitaries, see, e.g., \cite{haah2013}. 

\paragraph{{\bf Example.}} In the setting of the previous example, let us provide an elementary $\lm$-unitary that maps the product state Hamiltonian $H_0=\sum_{x\in \IZ} X_x+h.c.$, which corresponds to a stabilizer module generated by $\begin{pmatrix}1 \\
    0 
\end{pmatrix}$, to the cluster Hamiltonian \eqref{cluster}:
\begin{align*}
    \begin{pmatrix}1 & x^{-1}+x^{+1}\\
    0 & 1 
\end{pmatrix}\,.
\end{align*} This example demonstrates that there exists a shallow Clifford circuit preparing the cluster state from the product state.

Together with shallow circuits, there is another important class of CQCA which is given by generalized translations (we shall equally use the term ``shifts''). A translation along one of the lattice directions is obviously a CQCA, yet it cannot be presented by a finite number of layers of Clifford circuits. Indeed, if we try to do so, then we either loose the locality, or we end up with an infinite number of layers. Together with the ordinary translations, we are interested in generalized translations -- $\lm$-unitaries defined by a property that they do not mix the generalized Pauli $X$ and Pauli $Z$ operators. In other words, a generalized translation consists of a simultaneous translation and a mixing of species of qudits.  Generalized translations correspond to the subgroup of ``hyperbolic'' $\lm$-unitaries defined in \eqref{def:hyperbolicmap}.

\section{Algebraic preliminaries}\label{sec:algebraicintro}
In this section we give necessary definitions, set notation and conventions, and review some relevant theorems.

\subsection{Notations and conventions}
Throughout this paper, by $R$ we shall mean a regular commutative ring with involution where $2$ is invertible. The involution of $a\in R$ is denoted by $\bar a$. Here and throughout the paper we shall assume without mention that $R$-modules are finitely-generated projective (FGP). For an $R$-module $\L$, we define its dual as $\L^*\coloneqq \mbox{Hom}_R(\L,R)$ equipped with the following module structure: 
\begin{align*}
    R\times \L^*\to \L^*\,,\quad (a, Z)\mapsto \bar a \,Z\,.
\end{align*} 
The canonical pairing between $\L$ and $\L^*$ is the simple evaluation:
\begin{align*}
   \L\times \L^*\to R\,,\quad (X, Z)\mapsto Z(X)\,.
\end{align*} 

We adopt matrix notation for maps defined on direct sums of modules. For example, a module map $\L\oplus \mbox{M}\to \mbox{N}\oplus \mbox{Q}$ comprises of four maps: $\alpha\in \mbox{Hom}_R(\mbox{L}, \mbox{N})$, $\beta\in \mbox{Hom}_R(\mbox{M}, \mbox{N})$, $\gamma\in \mbox{Hom}_R(\mbox{L}, \mbox{Q})$, and $\delta\in \mbox{Hom}_R(\mbox{M}, \mbox{Q})$, which we denote by 
$\scaleto{
\begin{pmatrix} \alpha & \beta \\ \gamma & \delta \end{pmatrix}}{20pt}$. The dual to $\alpha\in \mbox{Hom}_R(\mbox{L}, \mbox{N})$ is the map $\alpha^*\in \mbox{Hom}_R(\mbox{L}^*, \mbox{N}^*)$ defined by
\begin{align*}
    \alpha^*(g)(X)=g(\alpha (X))\,\;\; \mbox{for any } X\!\in\! \L\, \mbox{ and }\, g\!\in\! \mbox{N}^*\,.
\end{align*} 
Similarly, the dual to $\scaleto{
\begin{pmatrix} \alpha & \beta \\ \gamma & \delta \end{pmatrix}}{20pt}$ is the map $\scaleto{
\begin{pmatrix} \alpha^* & \gamma^* \\ \beta^* & \delta^* \end{pmatrix}}{20pt}$. 

Note that the double dual of any FGP module is canonically isomorphic to itself: we use this canonical isomorphism to identify $\L^{**}$ and $\L$. A $\pm$hermitian form on an FGP $R$-module $\L$ is an element $\phi\in \mbox{Hom}_R(\L,\L^*)$ such that $\phi^*=\pm \phi$. The more familiar notion of a $\pm$hermitian form as a pairing corresponds to our definition through the following identification. A $\pm$hermitian form $\phi$ corresponds to the pairing $\tilde{\phi}$ via 
$\tilde \phi(X_1,X_2)=\phi (X_1)(X_2)$ such that $\tilde \phi(X_1,X_2)=\pm\overline{\tilde \phi(X_2,X_1)}$ and $\tilde \phi(X_1,aX_2)= a\tilde\phi(X_1,X_2)$. We denote the { \bf space of $+$hermitian forms} on $\L$ (on $\L^*$) by $\sl$ (by $\sls$). We say that a $\pm$hermitian form $\phi$ is non-degenerate if it is an isomorphism of $R$-modules.

\subsection{\texorpdfstring{$\l^{\pm}$}{Lg}-unitaries}\label{sec:unitaries}
For an FGP $R$-module $\L$, we define two modules $\h^{\pm}(\L)=\L \oplus \L^*$ equipped with the $\pm$hermitian form:
\begin{align}
    \lambda^{\pm}=
    \begin{pmatrix} 0 & 1 \\ \pm1 & 0 \end{pmatrix}\,,
\end{align}
such that the associated pairing is given by 
\begin{align*}
    \tilde \lambda^{\pm}((X_1,Z_1),(X_2,Z_2))=Z_1(X_2)\pm Z_2(X_1)\,.
\end{align*}
The group of module automorphisms of $\L\oplus \L^*$ preserving $\l ^{\pm}$ is called the group of {\bf $\lambda^{\pm}$-unitaries} and denoted by $\hp^{\pm}(\L;R)$. In the case of $\lm$-unitaries, there is a special subgroup generated by the following automorphisms:
\begin{align}\label{def:elun}
    \E_0(q_0)=\begin{pmatrix} 1 & 0 \\ q_0 & 1 \end{pmatrix}\,,\quad     \E_1(q_1)=\begin{pmatrix} 1 & q_1 \\ 0 & 1 \end{pmatrix}
\end{align} 
where $q_0$ is a $+$hermitian form on $\L$ and $q_1$ is a $+$hermitian form on $\L^*$. Such $\l^{-}$-unitaries are called elementary; they generate the group of {\bf elementary $\l^{-}$-unitaries}, denoted by $\ehp^{-}(\L;R)$\footnote{Note, our definition of elementary $\l^{\pm}$-unitaries differs from that of \cite{Haah2022}.}. The group of elementary $\l^{+}$-unitaries is generated by the transformations of the form \eqref{def:elun} with $q_0$ and $q_1$ being anti-hermitian forms and a new class of automorphisms of the form $\scaleto{\begin{pmatrix}
    0 & 1\\
    1 & 0
\end{pmatrix}}{15pt}$.

Another important subgroup of $\hp^{-}(\L;R)$ comes from automorphisms of $\L$, and is defined through the ``hyperbolic" group homomorphism:
\begin{align}
\label{def:hyperbolicmap}
    \h^{\pm}: \gl(\L;R)&\to \hp^{\pm}(\L;R)\,,\\
    a&\mapsto \begin{pmatrix} a & 0 \\ 0 & (a^{*})^{-1} \end{pmatrix}\,.
\end{align}

\subsection{Hermitian K-theory}
Let us give a minimal review of some facts from hermitian K-theory. We recommend \cite{KBook} as the standard reference on algebraic $\K$-theory.  Let $\mbox{P}(R)$ be the monoid of isomorphism classes of FGP $R$-modules with the direct sum as the binary operation. The ordinary $\K_0(R)$ is the abelian group which is obtained as the Grothendieck completion of $\mbox{P}(R)$. If $f$ is a ring homomorphism, then $\K_0(f)$ is a homomorphism of abelian groups induced by the extension of scalars.
In other words, $\K_0$ is a functor from the category of rings to the category of abelian groups.

In order to define $\K_1$, we first introduce the stable general linear group of $R$. There is an inclusion of groups:
\begin{align*}
     \gl(R^n;R)\to \gl(R^{n+1};R)\,,\quad u\mapsto u\oplus 1\,.
\end{align*}
The direct limit of these inclusions is called the stable general linear group of $R$ and denoted by $\gl(R)$. We define 
\begin{align*}
    \K_1(R)\coloneqq\gl(R)/[\gl(R),\gl(R)]
\end{align*}
where $[\gl(R),\gl(R)]$ is the subgroup of $\gl(R)$ formed by all possible commutators. Therefore, $\K_1(R)$ is an abelian group and comprises the functor $\K_1$. The Whitehead lemma provides an alternative description of $\K_1(R)$ as the quotient group $\gl(R)/\mbox{EGL}(R)$ where $\mbox{EGL}(R)$ is the stable group of elementary automorphisms. A useful property of $\K_1(R)$ that we shall use frequently is that, for any free $R$-module $\L$, the class of $a\in \gl(\L;R)$ in $\K_1(R)$ does not depend on the isomorphism $\L\cong R^{n}$ for some $n$, see Lemma 1.6 in \cite{KBook}.

While the algebraic K-theory of a ring concerns with FGP modules and automorphism groups of such modules, a hermitian K-theory of a ring with involution concerns with FGP modules equipped with a non-degenerate $\pm$hermitian form and module automorphisms preserving the form. We recommend \cite{KaroubiPeriodicity} as a general reference for the hermitian K-theory. Instead of the monoid of FGP modules $P(R)$, we now consider the monoid of FGP modules equipped with a non-degenerate $\pm$hermitian form, which we denote by $Q^{\pm}(R)$. The hermitian version of $\K_0(R)$ is the Grothendieck-Witt group of $R$, denoted by $GW^{\pm}(R)$, which is the Grothendieck completion of $Q^{\pm}(R)$. In the complete analogy with the stable group $\gl(R)$ we define stable $\l^{\pm}$-unitary groups, denoted by $\hp^{\pm}(R)$, and stable elementary $\l^{\pm}$-unitary groups, denoted by $\ehp^{\pm}(R)$. The analog of $\K_1(R)$ in the $\pm$hermitian case is the abelian group $ \hp^{\pm}(R)/\ehp^{\pm}(R)$.

The map $\L\to \h^{\pm}(\L)$ induces a homomorphism of groups $\K_0(R)\to GW^{\pm}(R)$, while map \eqref{def:hyperbolicmap} induces the homomorphism $\gl(R)\to \hp^{\pm}(R)$. Using the former hyperbolic map, we define the Witt group of non-degenerate $\pm$hermitian forms as follows:
\begin{align}\label{def:Wittgroup}
    \switt^{\pm}(R)\coloneqq\mbox{coker}(\h^{\pm}:\K_0(R)\to GW^{\pm}(R))\,.
\end{align}

 We shall be interested in the group of $\l^{\pm}$-unitaries modulo elementary and hyperbolic ones, which we denote by $\umod^{\pm}(R)$:
\begin{align}\label{def:UnitariesModElemModHyperb}
\umod^{\pm}(R)\coloneqq \hp^{\pm}(R)/\ehp^{\pm}(R)\;\h^{\pm}(\gl(R))\,.
\end{align}
This group is abelian, see Proposition 3.12 of \cite{Haah2022} for a proof, and $\umod^{\pm}$ is a functor from the category of rings to the category of abelian groups. 

\subsection{Polynomial extensions} For a given ring $R$, there are two extensions that will be relevant below: a polynomial extension $R[T]$ and Laurent polynomial extension $R[T,T^{-1}]$. Recall, $R[T]$ is the ring of polynomials in $T$ with coefficients in $R$ and where involution acts trivially on $T$:
\begin{align}
    \overline{{\sum_{n=0}^{N}}a_n\,T^{n}}=\sum_{n=0}^{N}\bar {a}_n\,T^{n}\,.
\end{align}
The $R[T,T^{-1}]$ is the ring of Laurent polynomials in $T$ with coefficients in $R$ and where involution acts on $T$ by inversion
\begin{align}
    \overline{{\sum_{n=-N}^{N}}a_n\,T^{n}}=\sum_{n=-N}^{N}\bar {a}_n\,T^{-n}\,.
\end{align}
There are ring homomorphisms given by evaluation at $1$:  
\begin{align*}
    ev_1:R[T]\to R\,,\quad ev_1:R[T,T^{-1}]\to R
\end{align*}
 while $R[T]$ also admits evaluation at $0$. As functors, $\K_0$ and $\K_1$ associate group homomorphisms to these ring homomorphisms -- this is a subject of fundamental theorems of (hermitian) K-theory. A fundamental theorem for $\K_0$- and $\K_1$-functors, in particular, states that for any regular ring $R$ there are isomorphisms $\K_0(R[T])\cong \K_0(R)$ and $\K_1(R[T])\cong \K_1(R)$, see Chapter II and Chapter III of \cite{KBook} respectively -- these statements are also known as homotopy-invariance theorems. The homotopy invariance of $GW^{\pm}(R)$ for a regular $R$ is Theorem 1.2 of \cite{KaroubiPeriodicity}. The homotopy invariance of $\umod^{\pm}(R)$ for a regular $R$ with $2$ invertible follows from Corollary 0.8 of \cite{KaroubiPeriodicity} and \cite{Karoubi1980}.\footnote{The homotopy invariance of $\umod^{-}(R)$ for regular rings with trivial involution and where $2$ is invertible is proven in Appendix D of \cite{BL}.} 

\subsection{Lagrangians}\label{sec:lagrangians}
Let $\La$ be a submodule of $\h^{-}(\L)$: it is called a {\bf Lagrangian} if it is:

1) isotropic: $\tilde{\lambda}^{-}|_{\scaleto{\La}{5pt}}=0$,

2) coisotropic: $\La^{\perp}\subset\La$, where $\La^{\perp}$ is the orthogonal complement to $\L$ with respect to $\tilde \l^{-}$,

3) direct summand in $\h^{-}(\L)$.

There is the standard Lagrangian in $\hm(\L)$, the submodule of the form $(\L,0)$, for brevity, we shall simply denote it by $\L\subset \hm(\L)$. The {\bf set of Lagrangians} in $\h^{-}(\L)$ is denoted by $\LL(R)$ -- this is a pointed set with the distinguished point $\L$. In what follows, we will be interested in the set of Lagrangians up to stabilization. Let us define the {\bf Lagrangian Grassmannian} $\CL(R)$ as a direct limit of inclusions:
\begin{align*}
    \CL_{R^n}(R)&\to \CL_{R^{n+1}}(R)\,\\
    \La&\mapsto \La\oplus R\,
\end{align*}
Thus defined $\CL(R)$ is a pointed set with the structure of an abelian monoid with the direct sum as the binary operation. 

Below we shall analyze the orbits of the group of elementary $\lm$-unitaries. Elementary $\lm$-unitaries are closely related to the notion of transversality of Lagrangians. Let $\La_1,\La_2\in \LL(R)$, we say that these Lagrangians are transversal if their set-theoretic union is the whole $\hm(\L)$. We denote this relation by $\pitchfork$:
\begin{align*}
    \La_1\pitchfork \La_2 \Longleftrightarrow\La_1 +\La_2 = \hm(\L)\,.
\end{align*}
Obviously, $\L\pitchfork \L^*$. It is easy to see that any Lagrangian transverse to $\L^*$ is a graph of some $+$hermitian form $q_0$ on $\L$:
\begin{align}
    \La \pitchfork \L^{*}\Longleftrightarrow \La=\begin{pmatrix}
    1 & 0\\
    q_0 & 1\end{pmatrix} \cdot \L\,.
\end{align}
and, similarly, any Lagrangian transverse to $\L$ is a graph of some $+$hermitian form $q_1$ on $\L^*$:
\begin{align}
\label{L*transvesal}
    \La \pitchfork \L\Longleftrightarrow \La=\begin{pmatrix}
    1 & q_1\\
    0 & 1\end{pmatrix} \cdot \L^*\,.
\end{align}
Further, we shall use the following elementary observation:
\begin{align}
\begin{pmatrix}
    1 & 0\\
    q_0 & 1\end{pmatrix} \cdot \L\pitchfork \L \Longleftrightarrow q_0 \,\mbox{\;is non-degenerate}\,,
\end{align}
and analogous statement for $\L^*$.

\textbf{Serre-Suslin-Swan theorem.} We aim to apply our results to the rings of Laurent polynomials with coefficients in a prime field $\IF_p$ for some $p>2$. For such rings a powerful theorem by Serre-Suslin-Swan holds \cite{Suslin1977,Swan1978}. This theorem states that all finitely-generated projective modules over $\ifpd$ are free. In other words, $\K_0(\ifpd)\cong \IZ$. One might also relate this theorem with the Quillen-Suslin theorem in algebraic geometry. Any finitely-generated module over $\IR$ is isomorphic to $R^{n}$ for some natural number $n$. 
Let us also notice that since any Lagrangian module  $\La\in \LL(\ifpd)$ is free, there always exists a $\lm$-unitary $\Phi$ such that $\La=\Phi \cdot \L$. This is a well-known theorem, see, e.g., Proposition 3.10 of \cite{Haah2022}.

\section{Paths of Lagrangians}\label{sec:pathsoflagrnagians}

 In this section we discuss paths of Lagrangians and connected components of the space of Lagrangians. We fix a ring $R$ which we assume to be regular, commutative, equipped with an involution, and such that $2$ is invertible.

The ring extension $R[T]$ comes with the evaluation homomorphisms $ev_0$ and $ev_1$; the induced action on modules, the extension of scalars, maps $\ifpd[T]$-modules to $\ifpd$-modules -- for brevity, we shall call this action simply as evaluation. Due to the homotopy invariance of $\K_0$, we can morally think of FGP modules over $R[T]$ as of ``paths of modules'' over $R$. We say that two Lagrangians $\La_1,\La_2\in \LL(R)$ are homotopic if there exists a ``path" $\alpha \in \mathcal{L}_{\scaleto{R[T]\otimes_R \L}{7pt}}(R[T])$ such that $ev_0^{*}(\alpha)=\La_1$ and $ev_1^{*}(\alpha)=\La_2$. Let us notice that composition of algebraic paths of Lagrangians is not defined.  discussing paths of modules, we use the module $R[T]\otimes_R \L$ obtained by extension of scalars induced by the inclusion $R\to R[T]$. For brevity, we shall denote $\L[T]\coloneqq R[T]\otimes_R \L$.

Algebraic homotopies used in this text are continuous in the following sense. By the Yoneda lemma, each algebraic homotopy defines a morphism from the affine line $\mathbb{A}^1_{R}$  to the Lagrangian Grassmannian over $R$. In other words, any such a homotopy defines a path between $R$-points and is continuous in the Zariski topology\footnote{We thank G. Shuklin for this remark.}. Observe that this data can also be interpreted as an $\mathbb{A}^1$-homotopy \cite{voevodski}.

Since our goal is to study $\lm$-unitaries, we shall be mostly working with free $R$-modules. Our first result relates the connected component of $\LL(R)$ with paths of $\lm$-unitaries.
\begin{proposition}
    Let $\L$ be a free $R$-module and $\alpha\in \CL_{L}(R[T])$ be such that $\alpha(0)=L$, then there exists a free $R$-module $\L'$ and $\Phi\in \hp^{-}(\L\oplus \L';R[T])$ with $\Phi(0)=1$ such that 
    \begin{align*}
        \alpha\oplus \L'[T] =\Phi\cdot (\L[T]\oplus \L'[T])\in \CL_{L}(R[T])\,.
    \end{align*}
\end{proposition}
\begin{proof}
The homotopy invariance of $\K_0$ states that the evaluation map $\K_0(R[T])\to K_0(R)$ is an isomorphism, which in our case means that the class of $\alpha$ in $\K_0(R[T])$ is the same as the class of $\L$ in $\K_0(R)$. Since $\L$ is free, its class in $\K_0(R)$ is trivial, and the isomorphism $\K_0(R)\to K_0(R[T])$ is given by the extension of scalars, $\alpha\oplus R[T]^m$ must be free for some natural $m$. 
 Therefore, due to Proposition 3.10 of \cite{Haah2022}, there exists $\Phi\in \hp^{-}(\L[T]\oplus R[T]^m;R[T])$ such that $\alpha\oplus R[T]^m=\Phi'\cdot(\L[T]\oplus R^m[T])$. In case $\Phi'(0)\neq 1$, we can always redefine $\Phi(T)=\Phi'(T) \Phi'^{-1}(0)$ such that $\Phi(0)=1$.
\end{proof}
We notice that due to the Serre-Suslin-Swan theorem \cite{Suslin1977,Swan1978}, all FGP modules over \linebreak $R=\ifpd$ are free, in particular, $\alpha$ is free, and for such rings the statement of this proposition holds without the stabilization.

 The following lemma demonstrates that a path of $\l^{-}$-unitaries that starts at the identity is a composition of ``shifts'' and elementary $\lm$-unitaries up to stabilization. 
\begin{proposition}
    Let $\L$ be a free $R$-module and $\Phi\in \hp^{-}(\L[T];R[T])$ with $\Phi(0)=1$, then there exists a free module $\L'$, an element $\Psi \in \ehp^{-}(\L[T]\!\oplus\!L'[T];R[T])$, and an element $\Gamma\in \gl(\L[T]\oplus L'[T];R[T])$ such that $\Phi=\Psi \mbox{H}^{\,-}(\Gamma)$.
\end{proposition}

\begin{proof}
We use the homotopy invariance of $\umod^{-}(R)=\hp^{-}(R)/\ehp^{-}(R)\h(\gl(R))$: $\umod^{-}(R[T])\cong \umod^{-}(R)$ where isomorphism $\umod^{-}(R[T])\to\umod^{-}(R)$ is induced by the evaluation map. Therefore, if a path of $\lm$-unitaries  $\Phi$ evaluates to the identity, then it must represent the trivial element in $\umod^{-}(R[T])$. A lift of the trivial element of $\umod^{-}(R[T])$ is a $\l^{-}$-unitary which is a composition of an elementary $\lm$-unitary and a shift, up to stabilzation.
\end{proof}

As a corollary from the lemmas above, we obtain a characterization of the connected component of $\CL_L(R)$. 

\begin{corollary}\label{Corollary:CLtoESp}
Let $\L$ be a free $R$-module and $\alpha\in \CL_L(R[T])$ be such that $\alpha(0)=L$ and $\alpha (1)=\Lambda$. Then, there exists a free $R$-module $L'$ such that 
\begin{align}
    \alpha \oplus L'[T]=\Psi (L[T]\!\oplus\!L'[T])
\end{align}
for some $\Psi\in \ehp^{-}(\L[T]\!\oplus\!L'[T];R[T])$.

\end{corollary}
In other words, a path of Lagrangians beginning at the base point can be presented as a path of elementary $\lm$-unitaries up to stabilization. Vice versa, if a Lagrangian is connected to the standard Lagrangian by an elementary $\lm$-unitary, then there is a homotopy between the two.

\begin{proposition}
\label{Prop:EsptoPath}
    Let  $\La\in \CL_{L}(R)$ be such that $\La=\Psi \cdot \L$ for some $\Psi\in \ehp^{-}(\L;R)$. Then, there exists a path $\alpha\in \CL_L(R[T])$ such that $\alpha(0)=L$ and $\alpha (1)=\Lambda$
\end{proposition}
\begin{proof}
    By the assumption, $\Psi= E_0(q_0)\cdot E_1(q_1)  \cdots  E_0(q_{n-1})\cdot E_1(q_n)$ where $q_1,\dots, q _n$ are +hermitan forms defined on $L$ and $L^*$ alternatingly and where some of them can be zero-forms. The desired path is given by $\alpha =E_0(T q_0)\cdot E_1(Tq_1)  \cdots  E_n(Tq_n)\cdot \L$.
\end{proof}

 Combining Corollary \ref{Corollary:CLtoESp} and Proposition \ref{Prop:EsptoPath}, we obtain the following bijection of monoids:
\begin{align}\label{pi0LvsL/ESp}
    \pi_0\CL(R)\cong \CL(R)/\ehp^{-}(R)\,.
\end{align}

\subsection{Sturm sequences}
We have shown that any path of Lagrangians corresponds to an elementary $\lm$-unitary (in a highly non-unique way) which is encoded by a sequence of $+$hermitian forms. In what follows, we shall obtain a Lagrangian transversal to any Lagrangian obtained from the standard Lagrangian by an elementary $\lm$-unitary. This result will allow us to obtain a concise representation for any path of Lagrangians in $\CL(R)$.

We begin with a review of the technique of Sturm sequences of \cite{BL} adjusted for rings with involution. Let us fix a free $R$-module $\L$ and introduce the following notation: 
\begin{align*}
\L_n\coloneqq \begin{cases}
    \L \;\;,\quad n\equiv 0\;\mbox{mod}\;2\,\\
    \L^* \,,\quad n\equiv 1\; \mbox{mod}\;2\,\\
\end{cases}\quad 
\E_n\coloneqq \begin{cases}
    \E_0 \,,\quad n\equiv 0\;\mbox{mod}\;2\,\\
    \E_1 \,,\quad n\equiv 1\; \mbox{mod}\;2\,\\
\end{cases}
\quad 
\sigma_n\coloneqq \begin{cases}
    1 \,,\quad n\equiv 0\;\mbox{mod}\;2\,\\
    \sigma \,,\quad n\equiv 1\; \mbox{mod}\;2\,\\
\end{cases}
\end{align*} 
where $\sigma: \hm(\L)\to \hm(\L^*)$ is a module isomorphism whose matrix representation is that of $\l^-$. A Sturm sequence of type $(m,n)$ is a sequence of $+$hermitian forms $(q_m,q_{m+1},\dots, q_n)$ where $q_n\in \mathcal{S}_{{\scaleto{\L}{5pt}}_n}$: such  a sequence is denoted by $\underline{q}$. Any Sturm sequence $\underline{q}$ defines an elementary $\lm$-unitary:
\begin{align*}
    \E(\underline{q})=\E_{m}(q_m)\cdots \E_{n}(q_n)\,.
\end{align*}
\begin{proposition}\label{prop:threetotwo}
    Let $\underline{q}$ be a Sturm sequence of type $(m,n)$ and $(x_{m-1},x_m,\dots ,x_n,x_{n+1})$ be a sequence of elements $x_i\in \L_i$ with $m-1\leqslant i\leqslant n+1$. Then, for any $m\leqslant k\leqslant n$ the following conditions are equivalent:

    1) The sequence $(x_{k-1},x_k,x_{k+1})$ satisfies
 \begin{align*}
        x_{k-1}+(-1)^k q_k(x_k)+x_{k+1}=0\,.
\end{align*}
    2) The sequence $(x_{k-1},x_k,x_{k+1})$ satisfies
 \begin{align*}
        (x_{k-1},x_k)=(-1)^{k-1}(\sigma_{k-1}\cdot E_k(q_k)\cdot \sigma_k^{-1})(x_k,x_{k+1})
\end{align*}
Here $(x_k,x_{k+1})$ is identified with the element of $x_k\oplus x_{k+1}\in \hm(\L_k)$.

\begin{proof}
    The proof is by elementary unpacking the definitions. 
\end{proof}

\end{proposition}
For a Sturm sequence $\underline{q}$ of type $(m,n)$, we define $\L_{m,n}\coloneqq \displaystyle{\oplus_{k=m}^n}\L_k$ and $S(\underline{q})$ be a $+$hermitian form on $L_{m,n}$ with the following matrix representation 
\begin{align}
\begin{pmatrix}
\scaleto{(-1)^mq_m}{9pt} & 1 & &  &  &  & & \\
1 & \scaleto{(-1)^{m+1}q_{m+1}}{9pt} & 1 & & &  & &\\
 & 1 &  &  &  & & &\\
 &  &  &  &  & & &\\
 &  &  &  &  & & &  1\\
  &   &  &  &   & & 1 & \scaleto{(-1)^nq_n}{9pt}  \\
\end{pmatrix}\,.
\end{align}

The following proposition gives a constructive completion to any Lagrangian obtained from the standard Lagrangian by an elementary $\lm$-unitary.

\begin{proposition}\label{Prop:SturmSeq}
Let $\L$ be a free $R$-module and let $\underline{q}=(q_0,\dots, q_{2n})$ and $\underline{q}'=(q_0,\dots, q_{2n-1})$ be a pair of Sturm sequences of type $(0,2n)$ and $(0,2n\!-\!1)$ respectively. Denote by $\La$ the Lagrangian defined by the first sequence:
\begin{align}
    \La=\E(\underline{q})\cdot L\,
\end{align}
then, the following Lagrangians are transverse
\begin{align}
    \begin{pmatrix}
    1 & 0 \\
    S(\underline{q}') & 1 
    \end{pmatrix}\cdot (L\oplus L_{1,2n-1}) \pitchfork \La \oplus L_{1,2n-1}\,.
\end{align}
\end{proposition}

\begin{proof}
    We begin with noticing that $\L_{1,2n-1}^{\perp}=\L_{0,2n-1}\oplus \L_0^*$ within $\hm(\L_{0,2n-1})$ (obviously, $\L_{1,2n}$ is not a Lagrangian in $\hm(\L_{0,2n-1})$ as it is not coisotropic). Let us introduce the notation 
    \begin{align*}
        X(\underline{q}')\coloneqq \begin{pmatrix}
        1 & 0\\
        S(\underline{q}') & 1
    \end{pmatrix}\cdot \L_{0,2n-1}
    \end{align*} 
    and notice that $X(\underline{q})\pitchfork L_{1,2n-1}^{\perp}$, what follows from the concrete form of $S(\underline{q})$. 
    
    Let $i:\L_{1,2n-1}\to \L_{0,2n-1}$ be an inclusion and $i^{*}:\L_{0,2n-1}^*\to \L_{1,2n-1}^*$ be its dual. The transversality condition $X(\underline{q})\pitchfork L_{1,2n}^{\perp}$ implies that $i^*\circ S(\underline{q})$ is a surjective map. Let us denote the kernel of $i^*\circ S(\underline{q})$ by $\kappa\coloneqq \mbox{ker}(i^*\circ S(\underline{q}))$. An element $(x_0,\dots x_{2n-1})\in \L_{0,2n-1}$ belongs to $\kappa$ iff it is a solution of the following system
    \begin{align}\label{Kernel}
        x_{k}+(-1)^{k}q_k(x_k)+x_{k+1}=0\quad \mbox{for} \;1\leqslant k \leqslant 2n-1
    \end{align}

where we put $x_{2n+1}=0$. Algebraic system \eqref{Kernel} has rank $2n\!-\!1$ and we can choose $x_{2n-1}$ as an independent solution. Let $\rho: \kappa \to \L_{2n-1}$ be the isomorphism of ``solving system \eqref{Kernel}''. It is elementary to check that 
\begin{align*}
    X(\underline{q})\cap L_{1,2n-1}^{\perp}=\begin{pmatrix}
    1 & 0\\
    S(\underline{q}) & 1
\end{pmatrix}\cdot \kappa\,.
\end{align*}
Let $\mu$ be the projection $X(\underline{q})\cap L_{1,2n-1}^{\perp}\to  \L_{1,2n-1}^{\perp}/\L_{1,2n-1}$. We wish to track the image of $x_{2n-1}\in \L_{2n-1}$ under the sequence of maps 
\begin{align*}
    \L_{2n-1}\to \kappa \to X(\underline{q})\cap L_{1,2n-1}^{\perp}\to  \L_{1,2n-1}^{\perp}/\L_{1,2n-1}\cong \hm(\L_0)\,.
\end{align*}
We take an element $x_{2n-1}\in \L_{2n-1}$ and rewrite it through system \eqref{Kernel}. The composition $\mu\circ \begin{pmatrix}
    1\\
    S(\underline{q})
\end{pmatrix}\circ \rho^{-1}(x_{2n-1})$ is equal to $E(\underline{q}')\cdot \begin{pmatrix}
    x_{2n}\\
    0
\end{pmatrix}$ -- this is a result of consecutive applications of point $2)$ of Proposition \ref{prop:threetotwo}. Therefore, the image of $X(\underline{q}')\cap \L_{1,2n}$ in $\hm(L_0)$ is given by $E(\underline{q}')\cdot \L_{2n-1}$.

To this end, we use an elementary fact $E(\underline{q}')\cdot \L_{2n-1}\pitchfork E(\underline{q})\cdot \L_{2n}$ holding for any $+$hermitian form $q_{2n}$. We conclude the proof with the following observation
\begin{align*}
E(\underline{q}')\cdot \L_{2n-1}\pitchfork E(\underline{q})\cdot \L_{2n}  \Leftrightarrow X(\underline{q}')\pitchfork E(\underline{q})\cdot\L_{2n}\oplus  L_{1,2n-1}
\end{align*}
for which we essentially use the fact that $E(\underline{q}')\cdot \L_{2n-1}$ is the image of $\mu$.
\end{proof}

 This proposition associates a single $+$hermitian form defined on a bigger space to any Lagrangian obtained from the standard Lagrangian by a sequence of elementary $\lm$-unitaries. We immediately obtain a corollary from this proposition by using equivalence \eqref{L*transvesal}. 
 \begin{corollary}\label{cor:stabsturm}
Let $\L$ be a free $R$-module, $\underline{q}=(q_0,\dots, q_{2n})$ be a Sturm sequence of type $(0,2n)$, and $\La=\E(\underline{q})\cdot L$. Then, there exists a free $R$-module $\L'$, $+$hermitian form $X$ on $\L\oplus \L'$ and a $+$hermitian form $Y$ on $(\L\oplus \L')^{*}$  such that 
\begin{align}
\La\oplus \L'=\begin{pmatrix}
    1 & 0\\
    X & 1
\end{pmatrix}\begin{pmatrix}
    1 & Y\\
   0 & 1
\end{pmatrix}\cdot(\L\oplus \L')^{*}\,.
\end{align}
 \end{corollary}
 Proposition \ref{Prop:SturmSeq} and Corollary \ref{cor:stabsturm} provide with a convenient presentation of a Lagrangian obtained from the standard one by a sequence of elementary $\lm$-unitaries, and, consequently, to any Lagrangian homotopy equivalent to the standard Lagrangian. This presentation will be useful for the analysis of loops of Lagrangians that we conduct in later sections.

\section{An interlude}\label{interlude}

Let us comment more on the correspondence between stabilizer modules and stabilizer Hamiltonians. Any concrete stabilizer Hamiltonian like \eqref{Hamiltonian} generates an isotropic submodule of the Pauli module. Since the set of generators for any module is highly non-unique, many stabilizer Hamiltonians may generate the same stabilizer module. In the opposite direction, if two stabilizer Hamiltonians generate the same module, then they are homotopy-equivalent, i.e., there exists a path of gapped Hamiltonians connecting them, see Proposition 2.1 of \cite{haah2013}, which is formulated for qubits but can be adjusted to qudits of any prime dimension. Thus, there is a 1-1 correspondence between the stabilizer modules and homotopy equivalence classes of stabilizer Hamiltonians. We denote the homotopy class of the Hamiltonian $H$ by $\{H\}$. 

It turns out that such properties as the ground state degeneracy are characteristics of the stabilizer module and not of a specific Hamiltonian \cite{haah2013}. We focus on special classes of Hamiltonians corresponding to Lagrangian stabilizer modules as they are particularly simple. Let $\La_1$ and $\La_2$ be a pair of Lagrangians in the Pauli module that can be connected by an elementary $\lm$-unitary, and let $\{H_1\}$ and $\{H_2\}$ be a pair of corresponding classes of Hamiltonians, then, we have an equivalence:
\begin{align*}
        \{H_1\}\xleftrightarrow{\text{CFDQC}} \{H_2\} \quad &\Longleftrightarrow \quad \La_1 \xleftrightarrow{\ehp^{-}} \La_2
    \end{align*}

In the previous section we have discussed the algebraic homotopy relation between Lagrangian submodules. We can think of such a homotopy as a time-dependent Hamiltonian interpolating between the end points of the path. Note that the time parametrizing the path is not on the same footing as the spatial coordinates. The algebraic homotopies are defined through the polynomial extension of the Pauli algebra. Monomials in the non-extended Pauli algebra generate lattice translations, while the new time variable generates translations only in one direction. We observed that if two Lagrangians are connected by an elementary $\lm$-unitary, then they are homotopy equivalent and vice versa up to stabilization. In other words, we have an equivalence
\begin{align*}
 %      \{H_1\}\xleftrightarrow{\text{LGU}} \{H_2\} \quad &\Longleftrightarrow \quad \La_1 \xleftrightarrow{\text{homotopy}} \La_2\\
        \{H_1\}\xleftrightarrow{\text{CFDQC}} \{H_2\} \quad &\Longleftrightarrow \quad [\La_1]\;\sim_h\; [\La_2]
    \end{align*}
where $[\La]$ stands for the stabilization of $[\La]$, i.e., the class of  Lagrangians obtained by adding all possible ancilla. In what follows, by a ``Hamiltonian", we shall mean  the homotopy class of Hamiltonians represented by this Hamiltonian.

Any stabilizer Hamiltonian corresponding to a Lagrangian stabilizer module can be disentangled by some Clifford QCA\footnote{Moreover, any invertible stabilizer Hamiltonian (having a unique ground-state on any lattice) can be disentangled with a Clifford QCA, see Theorem IV and Lemma IV.10 of \cite{haah2023nontrivial}. We thank J. Haah for pointing these results to us.}, which is unique up to a Clifford FDQC, according to the discussion in the concluding paragraph of Section \ref{sec:algebraicintro}. On the other hand, we can map a Clifford QCA to the image of the product Hamiltonian (the standard Lagrangian module) under this QCA. Therefore, there is a correspondence between stabilizer Hamiltonians generating Lagrangian modules modulo the action of Clifford circuits and Clifford QCA modulo Clifford circuits. According to \eqref{pi0LvsL/ESp} and the discussion above, the Lagrangian Grassmannian $\CL(R)$ modulo the action of Clifford circuits is isomorphic to the monoid of connected components of $\CL(R)$. Therefore, the study of phases of Clifford QCA boils down to the analysis of the connected components of the Lagrangian Grassmannian. 

Let us turn to the main subject of this paper, the loops of Clifford circuits that transform any Lagrangian submodule of the Pauli module to the same submodule. Loops of Clifford circuits correspond to loops of Lagrangians. Without loss of generality, we consider loops in the Lagrangian Grassmannian based on the distinguished point, the standard Lagrangian. We say that two loops are homotopy equivalent if there exists a two-parameter family of Lagrangians interpolating between the two loops. We shall demonstrate that the monoid of connected components of loops in the Lagrangian Grassmannian, modulo the action of generalized translations, is an abelian group. Due to the relation between QCA and Lagrangian submodules, we associate the latter group with the group of homotopy classes of loops of invertible states induced by Clifford circuits. 

\subsection{Classical Maslov index.}
The main tool of the next section is an algebraic generalization of the Maslov
index, which is well-known in symplectic geometry. Before we proceed to the algebraic version in the next section, we wish to demonstrate how the index works in the simplest case of a two-dimensional real phase space. As we pointed out in the introduction, the Pauli module is a direct analog of the phase space of Hamiltonian mechanics. The classical Maslov index of a loop of Lagrangian subspaces in a real phase space is a $\IZ$-valued invariant that can be thought of as the winding number of such a loop. 

Le us consider the setting of Hamiltonian mechanics: the Pauli algebra is $R=\IR$ and the Pauli module is a real phase space $\IR^{2n}$. It is a classical observation that $\mbox{U}_{n}/\mbox{O}_{n}$ acts freely and transitively on the space of all Lagrangian subspaces $\CL_{\IR^n}(\IR)$ of $\IR^{2n}$. Therefore, by picking a base-point, we obtain a diffeomorphism of manifolds
\begin{align}\label{diffeo}
\mbox{U}_{n}/\mbox{O}_{n}\xrightarrow{\sim} \CL_{\IR^n}(\IR)\,,\quad u\mapsto u\cdot (\IR^n\oplus 0)\,. 
\end{align}
Since the determinant of any matrix in $\mbox{O}_n$ is $\pm 1$, there is a well-defined map:
\begin{align}\label{detsq}
    \mbox{det}^2: \CL_{\IR^n}(\IR)\to \mbox{U}_1\,
\end{align}
whose fiber is the subspace of $\CL_{\IR^n}(\IR)$ corresponding to $(\mbox{det}\,u)^2=1$ under $\eqref{diffeo}$. Let us denote the fiber of \eqref{detsq} by $S\CL_{\IR^{n}}(\IR)$, then  we have a fibration 
\begin{align*}
    \mbox{SO}_n\to \mbox{SU}_n\to S\CL_{\IR^{n}}(\IR)
\end{align*}
The long exact sequence of homotopy groups for this fibration together with that for $S\CL_{\IR^{n}}(\IR) \to \CL_{\IR^n}(\IR) \to \mbox{U}_1$ gives us an isomorphism $\pi_1(\CL_{\IR^{n}}(\IR))\cong \IZ$.

The classical Maslov index is a tool calculating the degree of a loop of Lagrangians. Let us work out the simplest possible example of a loop of Lagrangians in the real two-dimensional phase space. The Lagrangians in $\IR^{2}$ are straight lines such that $\CL_{\IR}(\IR)\cong \mbox{U}_1/\mbox{O}_1 \cong \IR \mathbb{P}^1$. To make a connection with the generalized Maslov index considered in the next section, we restrict to loops parametrized by algebraic polynomials. As an example, we consider 
\begin{align}\label{maslovex}
    \alpha(T)=\begin{pmatrix}
        4(T+\frac{1}{\sqrt{2}})^3-6(T+\frac{1}{\sqrt{2}})^2+1\\
        12(T+\frac{1}{\sqrt{2}})^2-12(T+\frac{1}{\sqrt{2}})
    \end{pmatrix}\,,\quad T\in [0,1]\,.
\end{align}
The curve $\{\alpha(T)\,,\,T\in[0,1]\}\subset \IR^2$  defines a loop in $\CL_{\IR}(\IR)\cong \IR \mathbb{P}^1$ of degree one.

Let us describe an algebraic method of computing degrees of loops in $\CL_{\IR}(\IR)$ using Sturm functions \cite{BL,ghys2015signatures}. Let us consider $\alpha_P(T)$ of the form
$\scaleto{\begin{pmatrix}
    P(T)\\
    P'(T)
\end{pmatrix}}{20pt}$ where $P(T)$ is a polynomial of degree $m$ with simple roots and such that $P(0),P(1)\neq 0$. We use the Euclidean algorithm to divide $P(T)$ by $P'(T)$ which results in a sequence of residues $(q_1,\dots, q_m)$ where each element is a linear polynomial and $q_m\in \IR^{\times}$. One checks that the sequence of residues, also known as the Sturm sequence, linearizes  $\alpha_P(T)$ in the sense that
\begin{align*}
\begin{pmatrix}
    P(T)\\
    P'(T)
\end{pmatrix}=\begin{pmatrix}
    q_1 & -1\\
    1   & 0
\end{pmatrix}
\cdots
\begin{pmatrix}
    q_{m-1} & -1\\
    1   & 0
\end{pmatrix}\cdot
\begin{pmatrix}
    q_m \\
    0   
\end{pmatrix}\,.
\end{align*}
To this end, we define a bilinear form with the following matrix representation:
\begin{align*}
S=\begin{pmatrix}
q_1(T) & -1 & &  &  &  & & \\
-1 & q_2(T) & -1 & & &  & &\\
 & -1 &  &  &  & & &\\
 &  &  &  &  & & &\\
 &  &  &  &  & & &  -1\\
  &   &  &  &   & & -1 & q_m(T)  \\
\end{pmatrix}\,.
\end{align*}
Then, the degree of the loop parametrized by $\alpha_P(T)$ is given by $\frac{1}{2}\mbox{sign} (S(1)\oplus -S(0))$, where $\mbox{sign}$ is the signature of a bilinear form. Using this formula, we confirm that the loop defined  by \eqref{maslovex} indeed has degree one and this degree coincides with the Maslov index.

For the case of a real phase space, abstract algebraic loops of Lagrangians correspond to actual geometric loops like the ones considered above. In the next section we shall use an analog of the Maslov index for Lagrangian submodules inside the Pauli module. One should be aware that the Lagrangian Grassmannian for $R=\ifpd$ is not a manifold but a scheme, and all that we can do is put the Zariski topology on it. In this topology, the algebraic paths are actually continuous. 

 \section{Loops of lagrangians}\label{sec:loopsoflagrangians} 
Recall, a path of Lagrangians in $\LL(R)$ is an element in $\LL(R[T])$. The goal of this section is to analyze the space of loops in $\LL(R)$, which we denote by $\OLL(R)$:
\begin{align*}
    \OLL(R)\coloneqq \{\,\alpha \in \CL_{\scaleto{L[T]}{7pt}}(R[T])\;|\;\alpha(0)=\alpha(1)=\L\,\}\,.
\end{align*}

In a moment, we shall demonstrate a trick that shows that any loop of Lagrangians is homotopy-equivalent to the constant loop. One should be aware that the being homotopic within $\Omega \LL(R)$ is not the same as being homotopic within $\CL_{\scaleto{L[T]}{7pt}}(R[T])$. What we are going to show that for any loop $\alpha(T)\in\Omega \LL(R)\subseteq \CL_{\scaleto{L[T]}{7pt}}(R[T])$, there exists a homotopy to the constant loop $\L[T]$ within $\CL_{\scaleto{L[T]}{7pt}}(R[T])$. 

Let $g$ be a homomorphism which acts on polynomials by a substitution
\begin{align*}
g:R[T]\to R[T,\tau]\,,\quad   a(T)\mapsto a(T\tau)\,.
\end{align*} 
Then, the $R[T,\tau]$-module induced by $g$ from a loop $\alpha(T)\in \OLL(R)$ is given by $\aleph(T,\tau)=R[T,\tau]\otimes_g \alpha(T)$. Any such $\aleph(T,\tau)$ is a homotopy between $\alpha(T)$ and $L[T]$. Therefore, any $\alpha \in \OLL(R)$ satisfies Corollary \ref{Corollary:CLtoESp}. Using this observation and Proposition \ref{Prop:SturmSeq}, we obtain the following proposition. 
\begin{proposition}
\label{prop:Loopstoforms}
    Let $\alpha \in \Omega \CL_{L}(R)$, then there exists a free $R$-module $\L'$ and a $+$hermitian form $S(T)$ on $\L[T]\oplus L'[T]$ such that 
\begin{align}\label{Def:transversality}
    \alpha \oplus \L'[T] \pitchfork \begin{pmatrix}
        1 & 0 \\
        S(T) & 1\\
    \end{pmatrix}\cdot (\L[T]\oplus L'[T])\,.
\end{align}
\end{proposition}

Let us evaluate \eqref{Def:transversality} at $0$ and at $1$: we obtain a pair of relations 
\begin{align}\label{Def:transversalityeval}
    \L \oplus \L' \pitchfork \begin{pmatrix}
        1 & 0 \\
        S(0) & 1\\
    \end{pmatrix}\cdot (\L\oplus \L')\,,
    \quad \quad \quad
    \L \oplus \L' \pitchfork \begin{pmatrix}
        1 & 0 \\
        S(1) & 1\\
    \end{pmatrix}\cdot (\L\oplus \L')\,,
\end{align}
which imply that $S(0)$ and $S(1)$ are non-degenerate $+$hermitian forms. The assignment of $S(T)$ to a loop of Lagrangians is highly non-unique. We aim to show that the homotopy class of the form $S(0)\oplus -S(1)^{-1}$ completely determines the homotopy class of the loop. Before we discuss a homotopy invariant of a loop, let us introduce some definitions and prove auxiliary lemmas.

Given a free $R$-module $\L$, we denote the {\bf set of non-degenerate $+$hermitian forms on $\L\oplus \L^{*}$} by $\FL(R)$, which is a pointed set with the base point $\lambda ^{+}$. Similarly to the paths of Lagrangians we define paths of $+$hermitian forms in $\FL(R)$. We say that two forms $\phi,\psi \in \CF_{L}(R)$ are homotopy-equivalent, which we denote by $\phi \sim_h \psi$, if there exists $\Gamma(T) \in \FL(R[T])$ such that $\Gamma(0)=\phi$ and $\Gamma(1)=\psi$. The set of of connected components of $\FL(R)$ is denoted by $\pi_0 \FL(R)$.

The following lemma will be useful for comparing two $+$hermitian forms satisfying \eqref{Def:transversality} for the same loop of Lagrangians.
\begin{lemma}
    Let $\alpha \in \Omega \CL_{L}(R)$, $\L'$ be a free $R$-module and $S_1(T)$, $S_2(T)$ be a pair of $+$hermitian forms on $\L[T]\oplus \L'[T]$ such that 
    \begin{align}\label{pairtransverse}
        \alpha \oplus \L'[T] \pitchfork
        \begin{pmatrix}
            1 & 0\\
            S_l(T) & 1
        \end{pmatrix}\cdot (\L[T]\oplus \L'[T])\,,\quad l=1,2\,.
\end{align}
Then, 
\begin{align}\label{Eq:TrivSamePoint}
    S_2(0)\oplus -S_1(0)^{-1}\sim_h S_2(1)\oplus -S_1(1)^{-1} \,.
\end{align}\end{lemma}
\begin{proof}
Let us use Corollary \ref{cor:stabsturm} for $\alpha[T]$: there must exist a pair of $+$hermitian forms $Y_1(T)$ and $Y_2(T)$ on $\L[T]\oplus \L'[T]$ such that 
\begin{align}
\label{LoopEquiv}
    \alpha\oplus \L'[T]=
    \begin{pmatrix}
         1 &  0\\
         S_1(T) & 1
    \end{pmatrix}
        \begin{pmatrix}
         1 & Y_1(T)\\
         0 & 1
    \end{pmatrix}
    \cdot (\L[T]\oplus \L'[T])^*\,,\quad l=1,2\,.
\end{align}

Using \eqref{LoopEquiv} and \eqref{pairtransverse} with $l=2$, we obtain 
\begin{align*}
    \begin{pmatrix}
         1 & 0 \\
         -S_2(T) & 1
    \end{pmatrix}
   \begin{pmatrix}
         1 &  0\\
         S_1(T) & 1
    \end{pmatrix}
        \begin{pmatrix}
         1 & Y_1(T)\\
         0 & 1
    \end{pmatrix}
    \cdot (\L[T]\oplus \L'[T])^* \pitchfork (\L[T]\oplus \L'[T])\,
\end{align*}
by inverting $\scaleto{\begin{pmatrix}
    1 & 0 \\
    S_2(T) & 1
\end{pmatrix}}{20pt}$. It implies that $(S_1(T)-S_2(T))Y_1(T)+1$ is non-degenerate. We use this form to construct another non-degenerate form
\begin{align*}
      \begin{pmatrix}
         S_2(T)-S_1(T) & 1 \\
         1 & -Y_1(T)
    \end{pmatrix} \,.
\end{align*}
Indeed, $(S_1(T)-S_2(T))Y_1(T)+1$ is the determinant of such and it is invertible. We use this form to demonstrate that there is a homotopy
\begin{align*}
      \begin{pmatrix}
         S_2(0)-S_1(0) & 1 \\
         1 & Y_1(0)
    \end{pmatrix} \sim_h 
    \begin{pmatrix}
         S_2(1)-S_1(1) & 1 \\
         1 & Y_1(1)
    \end{pmatrix}\,.
\end{align*}
If we evaluate \eqref{LoopEquiv} at $0$ and $1$, we find that $Y_l(0)$ and $Y_l(1)$ with $l=1,2$ are non-degenerate and $Y_l(0)S_l(0)+1=0$, $Y_l(1)S_l(1)+1=0$ for $l=1,2$. Thus,
\begin{align*}
\begin{pmatrix}
         S_2(0)-S_1(0) & 1 \\
         1 & -S_1(0)^{-1}
    \end{pmatrix}
       \sim_h 
     \begin{pmatrix}
         S_2(1)-S_1(1) & 1 \\
         1 & -S_1^{-1}(1)
    \end{pmatrix}\,.
\end{align*}
Further, we use the following forms in $\FL(R[T])$
\begin{align*}
    &\begin{pmatrix}
        1 & TS_1(0)\\
        0 & 1
    \end{pmatrix}
        \begin{pmatrix}
        S_2(0)-S_1(0) & 1\\
        1 & -S_1^{-1}(0)
    \end{pmatrix}
        \begin{pmatrix}
        1 & 0\\
        TS_1(0) & 1
    \end{pmatrix}\,,\\
     &\begin{pmatrix}
        1 & TS_1(1)\\
        0 & 1
    \end{pmatrix}
        \begin{pmatrix}
        S_2(1)-S_1(1) & 1\\
        1 & -S_1^{-1}(1)
    \end{pmatrix}
        \begin{pmatrix}
        1 & 0\\
        TS_1(1) & 1
    \end{pmatrix}
\end{align*}
which demonstrate that 
\begin{align*}
    \begin{pmatrix}
        S_2(0)-S_1(0) & 1\\
        1 & -S_1^{-1}(0)
    \end{pmatrix}&\sim_h 
    \begin{pmatrix}
        S_2(0) & 0\\
        0 & -S_1^{-1}(0)
    \end{pmatrix}\,,\\
      \begin{pmatrix}
        S_2(1)-S_1(1) & 1\\
        1 & -S_1^{-1}(1)
    \end{pmatrix}&\sim_h \begin{pmatrix}
        S_2(1) & 0\\
        0 & -S_1^{-1}(1)
    \end{pmatrix}\,.
\end{align*}
Finally, we conclude that
\begin{align*}
        \begin{pmatrix}
        S_2(0) & 0\\
        0 & -S_1^{-1}(0)
    \end{pmatrix}\sim_h
    \begin{pmatrix}
        S_2(1) & 0\\
        0 & -S_1^{-1}(1)
    \end{pmatrix}\,.
\end{align*}
\end{proof}
\begin{lemma}
    Let $q$ be a non-degenerate $+$hermitian form on a free $R$-module $\L$. Then, 
    \begin{align}\label{Eq:TrivMas}
        q\oplus -q^{-1}\sim_h \l^{+} \,.
    \end{align}
\end{lemma}

\begin{proof} The proof essentially relies on the fact that $2$ is invertible in $R$. The following form delivers a concrete homotopy
\begin{align*}
    \begin{pmatrix}
        1 & -T\frac{q}{2}\\
        0 & 1 
    \end{pmatrix}
    \begin{pmatrix}
        1 & 0\\
        Tq^{-1} & 1 
    \end{pmatrix}
     \begin{pmatrix}
        q & 0\\
        0 & -q^{-1} 
    \end{pmatrix}
        \begin{pmatrix}
        1 & Tq^{-1}\\
        0 & 1 
    \end{pmatrix}
    \begin{pmatrix}
        1 & 0\\
        -T\frac{q}{2} & 1 
    \end{pmatrix}\,.
\end{align*}
\end{proof}

The previous lemma characterizes the connected component of the base-point in $\FL(R)$. Let us define the stable space $\CF(R)$ as the direct limit of inclusions of pointed sets
\begin{align*}
    \CF_{\scaleto{R^n}{6.5
    pt}}&\longrightarrow \CF_{\scaleto{R^{n+1}}{7pt}}\,,\\
    \phi &\longmapsto \phi \oplus \l^{+}\,.
\end{align*}
 Such a defined $\CF(R)$ is an abelian monoid.
Given a $\phi \in \FL(R)$ for some free $R$-module $\L$, its image in $\pi_0\CF(R)$ is denoted by $\st[\phi]$.

\subsection{Maslov index} Let $\alpha\in \Omega \LL(R)$ be a loop of Lagrangians and let $\L'$ and $S(T)$ be some free $R$-module and some $+$hermitian form respectively satisfying Proposition \ref{prop:Loopstoforms}, then we define the {\bf Maslov index of a loop} as 
\begin{align*}
    \mas(\alpha)\coloneqq \st [S(1)\oplus -S(0)^{-1}]\,
\end{align*}

\begin{theorem}\label{Thm:main}

The Maslov index does not depend on the choice of $S$ and $\L'$. There is an isomorphism of abelian monoids
\begin{align}\label{eq:MainThmStatement}
    \pi_0\Omega\CL(R)\cong \pi_0 \CF(R)\,.
\end{align}
\end{theorem}
\begin{proof}
   For the first part, assume $S_1(T)$ and $S_2(T)$ are two non-degenerate forms satisfying condition \eqref{Def:transversality}. We have the following chain of equalities within $\pi_0 \CF(R)$:
\begin{multline}
    \st[S_1(1)\oplus -S_1(0)^{-1}] \stackrel{(\ref{Eq:TrivMas})}{=} \st[S_1(1)\oplus -S_1(0)^{-1}\oplus S_2(0)\oplus -S_2(0)^{-1}]=\\
    \st[S_1(1)\oplus -S_2(0)^{-1}\oplus S_2(0)\oplus -S_1(0)^{-1}]\stackrel{(\ref{Eq:TrivSamePoint})}{=}
    \st[S_2(1)\oplus -S_2(0)^{-1}\oplus S_1(1)\oplus -S_1(1)^{-1}]\stackrel{(\ref{Eq:TrivMas})}{=}\\
    \st[S_2(1)\oplus -S_2(0)^{-1}]
\end{multline}
Thus, any $+$hermitian form on $\L[T]\oplus \L'[T]$ that satisfy \eqref{Def:transversality} gives the same Maslov index. 
Further we notice that adding a constant form does not change the Maslov index. Assume that
\begin{align}
    \alpha \oplus \L'\oplus \L'' \pitchfork \begin{pmatrix}
        1 & 0 \\
        S(T)\oplus q & 1\\
    \end{pmatrix}\cdot (\L\oplus L'\oplus \L'')\,,
\end{align}
then $\st[S(1)\oplus q \oplus - S(0)^{-1}\oplus -q]=\st[S(1) \oplus S(0)^{-1}]$ according to Eq. \eqref{Eq:TrivMas}. We conclude that the Maslov index is well-defined. 

Next, let us demonstrate that the Maslov index is only sensitive to the homotopy class of loops of Lagrangians. Let $\aleph(T,\tau)$ be a homotopy between loops $\alpha(T),\beta(T)\in \Omega\CL_L(R)$. In other words, $\aleph(T,\tau)$ is an element of $\CL_{L}(R[T,\tau])$ such that $\aleph(T,0)=\alpha(T)$ and $\aleph(T,1)=\beta(T)$, while $\aleph(0,\tau)=\aleph(1,\tau)=\L[\tau]$. Such a homotopy also satisfies Proposition \ref{Prop:SturmSeq} and there exists a free $R$-module $\L'$ and a $+$hermitian form $S(T,\tau)$ such that 
\begin{align}
    \aleph(T,\tau)\oplus \L'[T,\tau]\pitchfork \begin{pmatrix}
        1 & 0\\
        S(T,\tau) & 1
    \end{pmatrix}\cdot (\L[T,\tau]\oplus \L'[T,\tau])\,.
\end{align}
We notice that $S(0,\tau)$ and $S(1,\tau)$ are non-degenerate and we use
\begin{align*}
    S(1,\tau)\oplus -S(0,\tau)^{-1}
\end{align*}
as the homotopy between $\mas(\alpha)$ and $\mas(\beta)$. This demonstrates that the following diagram commutes     
\begin{center}
\begin{tikzcd}[column sep=scriptsize]
\Omega \LL(R) \arrow[dr] \arrow[rr, "\scaleto{\mas}{6pt}" ]
    & & \pi_0\CF(R) \\
& \pi_0\Omega \LL(R)  \arrow[ur]
\end{tikzcd}
\end{center}
     So far, we have shown that $\mas$ induces a map of sets $\pi_0\Omega\LL(R)\to \pi_0\CF(R)$, and, it is straightforward to check that the Maslov index induces a homomorphism of abelian monoids $\pi_0\Omega \CL(R)\to \pi_0\CF(R)$.

In order to construct the inverse, let us introduce a loop of Lagrangians parametrized by two non-degenerate $+$hermitian forms. Let $q_0,q_1\in \sl$, then we define the following loop 
\begin{align}\label{eq:loopparametrization}
    \hat l(q_0,q_1)=\begin{pmatrix} 1 & 0 \\ (1-T)q_0+Tq_1 & 1
        \end{pmatrix}
        \begin{pmatrix} 1 & (T-1)q_0^{-1}-Tq_1^{-1} \\ 0 & 1
        \end{pmatrix}\cdot L^*\in \Omega\LL(R)\,.
\end{align}
In particular, we have a map
\begin{align}
    l:\FL(R)\to \Omega \CL_{\scaleto{\L\oplus \L^*}{7pt}}(R)\,,\quad \phi \mapsto l(\phi)=\hat l(\l^{+},\phi)\,,
\end{align}
which preserves the base point: $\hat l(\l^{+},\l^{+})$ is the constant loop $\L[T]\oplus \L^*[T]$. Moreover, if $\chi(\tau)\in \FL(R[\tau])$ is the homotopy between $\phi$ and $\psi$, then $l(\chi(\tau))$ is the homotopy between $l(\phi)$ and $l(\psi)$. Further, if $\phi\in \FL(R)$ and $\psi\in \CF_{\scaleto{\L'}{6pt}}(R)$, then $l(\phi\oplus \psi)$ is the direct sum of loops in $\CL_{\scaleto{\L\oplus \L^*}{7pt}}(R)\oplus\CL_{\scaleto{\L'\oplus \L'^*}{7pt}}(R)$ (we assume that both $\L$ and $\L'$ are free). Therefore, $l$ induces a homomorphism of abelian monoids
\begin{align}
\label{Def:OppositetoMas}
    l: \pi_0 \CF(R)\to \pi_0 \Omega \CL(R)\,.
\end{align}

Note that $\mas (\hat l(q_0,q_1))=\st[q_1\oplus -q_0^{-1}]$. Therefore, any element in $\pi_0\Omega \LL(R)$ can be presented by the class of $\hat{l}(q_0,q_1)$ for some $q_0,q_1\in \sl$. Let us denote the class of $\alpha\in \Omega \LL(R)$ within $\pi_0\Omega \CL(R)$ by $\st[\alpha]$. We have the following chain of equalities:
\begin{align}
    \st[\,\hat{l}(q_0,q_1)]=\st[\,\hat l(q_0\oplus q_0^{-1},q_1\oplus -q_0^{-1})]=\st[\,\hat {l}(\l^{+},q_1\oplus -q_0^{-1})]=\st[\,l(q_1\oplus -q_0^{-1})]
\end{align}
which implies that \eqref{Def:OppositetoMas} is surjective. To this end, let us check that $\mas$ and $l$ are inverse to each other. We observe that $\l^{+}\sim_h -\l^{+}$ through the following homotopy
\begin{align}\label{def:homLptoLm}
    E_0(-T/2)\,E_1(T)\,E_0(-T)\,E_1(T/2)\,\l^{+}\,E_0(T/2)\,E_1(-T)\,E_0(T)\,E_1(-T/2)
\end{align}
Thus, $\mas \circ l(\phi)=\st[\phi\oplus -\l^{+}]=\st[\phi]$, while $l \circ \mas$ is identity. 
\end{proof}

This theorem can be seen as an incarnation of the fundamental theorem of hermitan K-theory \cite{Karoubi1980}, and it is crucial for the further discussion.

\subsection{More on \texorpdfstring{$\pi_0\CF(R)$}{Lg}}
The right hand side of \eqref{eq:MainThmStatement} is a cryptic abelian monoid and in this section we aim to give it more algebraic description. Recall the function $\hat l$ that we used for the proof of Theorem \ref{Thm:main}: it takes a pair of non-degenerate $+$hermitian forms as an input. We shall demonstrate that $(\pi_0 \CF)(R)$ can be related to a more tractable monoid consisting of such pairs of $+$hermitian subject to certain algebraic relations.

Let $\L$ be a free $R$-module and $q$ be a non-degenerate $+$hermitian form on $\L$; an isomorphism between the pairs $(\L;q)$ and $(\L';q')$ is an isomorphism of modules $a: \L\to \L'$ such that $q=a^{*}\circ q'\circ a$. We denote the isomorphism class of $(\L;q)$ by $[\L;q]$ and by $Q^{+lib}(R)$ we denote the abelian monoid of isomorphism classes of pairs $(\L;q)$ with the direct sum as the binary operation. The free Grothendieck-Witt group $GW^{+lib}(R)$ is the Grothendieck completion of $Q^{+lib}(R)$. Further, we consider triples of the form $(\L;q_1,q_2)$ where $\L$ is a free module and $q_1$,  $q_2$ are non-degenerate $+$hermitian forms on $\L$. An isomorphism between the triples $(\L;q_1,q_2)$ and $(\L';q_1',q_2')$ is an isomorphism of $R$-modules $a:\L\to \L'$ such that $q_1=a^*\circ q_1'\circ a$ and $q_2=a^*\circ q_2'\circ a$. We denote by $Q^{+lib}_1(R)$ the abelian monoid of isomorphism classes of triples $(\L;q_0,q_1)$ with the direct sum as the binary operation and by $GW^{+lib}_1(R)$ we denote its Grothendieck completion. The main object of our interest is a subgroup of $GW^{lib}_1(R)$  obtained by taking a quotient with respect to the ``gluing" relation 
    \begin{align}\label{def:bordrelation}
        [\L;q_0,q_1]+[\L;q_1,q_2]\sim_{b}[\L;q_0,q_2]\,.
    \end{align}
Thus we define $V(R)\coloneqq GW_1^{+lib}(R)/\sim_b$. We shall use the same notation for the quotient $[\L;q_0,q_1]\in V(R)$. A simple calculation 
\begin{align*}
    [\L;q_0,q_1]+[\L;q,q]=[\L;q_0,q]+[\L;q,q_1]+[\L;q,q]=[\L;q_0,q]+[\L;q,q_1]=[\L;q_0,q_1]
\end{align*}
demonstrates that $[\L;q,q]=0$ within $V(R)$. As a consequence, we have $[\L;q_0,q_1]=-[\L;q_1,q_0]$.

We aim to show that $V(R)$ is an avatar of $\pi_0\CF(R)$. In order to do that, let us first collect basic facts about $V(R)$. There is a ``bordism'' map:
    \begin{align}\label{def:dhat}
         d: GW_1^{+lib}(R)\longrightarrow GW^{+lib}(R)\,,\quad [\L;q_0,q_1]\longmapsto [\L;q_0]-[\L;q_1]\,
    \end{align}
     which is a homomorphism of abelian groups.
    The dimension of the free module is an obvious invariant of $[\L;q]\in GW^{+lib}(R)$ such that we have a group homomorphism 
    \begin{align}
        \mbox{dim}: GW^{+lib}(R)\longrightarrow \IZ\,,\quad [\L;q]\longmapsto  \dim\, \L\,.
    \end{align} 
    The kernel of $\dim$ is called the fundamental ideal $I^{lib}(R)\coloneqq \mbox{ker} (\mbox{dim} : GW^{+lib}(R)\to \IZ)$ of the free Grothendieck-Witt group. Homomorphism \eqref{def:dhat} induces a surjective homomorphism 
        \begin{align}
         \delta : V(R)\longrightarrow I^{lib}(R)\,.
    \end{align}

\begin{lemma}\label{lemma:nu}
   For any free $R$-module $\L$ and any non-degenerate $+$hermitian form $q$ on $\L$, there is a homomorphism of groups:
 \begin{align}\label{Def:nu}
    \nu_{[\scaleto{\L}{5pt};q]}: \gl(\L;R)\longrightarrow V(R)\,,\quad a\longmapsto [\L;q,a^*\sc q\sc a]\,.
 \end{align}
\end{lemma}
    \begin{proof}
    \begin{align*}
     [\L;q,a'^*\sc a^*\sc q\sc a \sc a']\stackrel{(\ref{def:bordrelation})}{=}  [\L;q,a'^*\sc q\sc a' ]+[\L; q\sc a',a'^*\sc a^*\sc q\sc a\sc a' ] =  [\L;q,a'^{*}\sc q\sc a' ]+[\L;q,a^*\sc q\sc a ] 
    \end{align*}
   \end{proof}

\begin{lemma}\label{lemma:K_1actsonV}
    There is a homomorphism of abelian groups
 \begin{align*}
     \nu:K_1(R)\longrightarrow V(R)
 \end{align*}
 which is induced by $\nu_{[\scaleto{\L}{5pt};q]}$, and which does not depend on $\L$ and $q$.
    \begin{proof}
    Let us compare $[\L;q,a^*\sc q\sc a]$ and $[\L';q',a'^*\sc q'\sc a']$ for different $a\in \gl(\L;R)$ and $a'\in \gl(\L';R)$. In order to do that, we introduce an auxiliary module $\L''$ and do the following algebra
    \begin{multline*}
        [\L;q,a^*\sc q\sc a ]=[\L\oplus \L'\oplus \L'';q\oplus q' \oplus q'',a^*\sc q\sc a \oplus q' \oplus q'' ]=\\
        [\L\oplus \L'\oplus \L'';q\oplus q' \oplus q'',(a^*\oplus 1 \oplus 1)\sc (q \oplus q' \oplus q'')\sc (a\oplus 1 \oplus 1) ]\,,
         \end{multline*}
         \begin{multline*}
        [\L';q',a'^*\sc q'\sc a' ]=[\L\oplus \L'\oplus \L'';q\oplus q' \oplus q'',q \oplus a'^* \sc q'\sc a' \oplus q'' ]=\\
        [\L\oplus \L'\oplus \L'';q\oplus q' \oplus q'',(1\oplus a'^*
        \oplus 1)\sc (q \oplus q' \oplus q'')\sc (1\oplus a'
        \oplus 1) ]
    \end{multline*}
    Using the previous lemma, we find that 
    \begin{multline*}
         [\L\oplus \L'\oplus \L'';q\oplus q' \oplus q'',(1\oplus a'^*
        \oplus 1)\sc (q \oplus q' \oplus q'')\sc (1\oplus a'
        \oplus 1) ]=\\
        - [\L\oplus \L'\oplus \L'';q\oplus q' \oplus q'',(1\oplus (a'^*)^{-1}
        \oplus 1)\sc (q \oplus q' \oplus q'')\sc (1\oplus a'^{-1}
        \oplus 1) ]
    \end{multline*}
    such that the difference is given by
    \begin{align*}
       [\L;q,a^*\sc q\sc a ]-[\L';q',a'^*\sc q'\sc a' ]= [\L\oplus \L'\oplus \L'';q\oplus q' \oplus q'',(a\oplus a'^{-1}
        \oplus 1)^*\sc (q \oplus q' \oplus q'')\sc (a\oplus a'^{-1}
        \oplus 1) ]\,.
    \end{align*}
We use the identity holding for any $a\in \gl(\L;R)$: 
\begin{align}\label{DiagIdentity}
    \begin{pmatrix}
        a & 0\\
        0 & a^{-1}
    \end{pmatrix}=
     \begin{pmatrix}
        1 & a\\
        0 & 1
    \end{pmatrix}
     \begin{pmatrix}
        1 & 0\\
        -a^{-1} & 1
    \end{pmatrix}
    \begin{pmatrix}
        1 & a\\
        0 & 1
    \end{pmatrix}
    \begin{pmatrix}
        1 & 0\\
        1 & 1
    \end{pmatrix}
     \begin{pmatrix}
        1 & -1\\
        0 & 1
    \end{pmatrix}
     \begin{pmatrix}
        1 & 0\\
        1 & 1
    \end{pmatrix}
\end{align}
If $a$ and $a'$ represent the same class in $\K_1(R)$, then they are equivalent up to stabilization and elementary automorphism: combination of this fact with identity \eqref{DiagIdentity} demonstrates that  $(a\oplus a'^{-1}
        \oplus 1)$ is an elementary automorphism if the auxiliary module $\L''$ is large enough. Thus, we have demonstrated that the difference
    \begin{align*}
        [\L;q,a^*\sc q\sc a ]-[\L';q',a'^*\sc q'\sc a' ]=
         [\L\oplus \L'\oplus \L'';q\oplus q' \oplus q'',e^*\sc (q \oplus q' \oplus q'')\sc e]
    \end{align*}
    for some elementary automorphism $e\in \egl(\L\oplus \L'\oplus \L'';R)$, which concludes the proof.
    \end{proof}
\end{lemma}
 
\begin{corollary}\label{Cor:Vfactorization}
    For any triple $[\L;q_0,q_1]\in V(R)$ and any $a_0,a_1\in \gl(R)$, the following relation holds
    \begin{align*}
        [\L;a_0^{*}\sc q_0\sc a_0,a_1^*\sc q_1\sc a_1]=[\L;q_0,q_1]+\nu\, (\Det\;( a_1\sc a_0^{-1}))
    \end{align*}
    where $\Det(a_1\sc a_0^{-1})$ is the image of $a_1\sc a_0^{-1}$ in $\K_1(R)$.
\end{corollary}

 Our immediate goal is to connect $V(R)$ with $\CF(R)$.
\begin{lemma}\label{lemma:CF/EGL}
    The abelian monoid $\CF(R)/\egl(R)$ is a group.
\end{lemma}
\begin{proof}
  Any element in $\CF(R)$ is the stabilization of some form $\phi\in \FL(R)$. For appropriate choice of a module, homotopy \eqref{Eq:TrivMas} evaluated at $T=1$ gives an elementary transformation, denote it by $e$, such that $e^*\sc (\phi \oplus -\phi^{-1})\sc e=\l^{+}\in \CF_{\scaleto{\L\oplus \L}{6pt}}(R)$. Therefore, the class of $ -\phi^{-1}$ is inverse to the class of $\phi$ in  $\CF(R)/\egl(R)$. 
\end{proof}

\begin{lemma}\label{Lem:homFtoV}
    The map 
    \begin{align*}
   \chi :\FL(R)\longrightarrow V(R)\,,\quad \phi\longmapsto [\L\oplus \L^*;\l^{+},\phi]
    \end{align*}
factors through the stabilization with the projection $\FL(R)\to \CF(R)/\egl(R) $ where $\egl(R)$ acts on $\CF(R)$ by conjugations. The following diagram commutes
    \begin{center}
\begin{tikzcd}[column sep=scriptsize]
\FL(R) \arrow[dr] \arrow[rr,"\scaleto{\chi}{4
pt}"  ]
    & & V(R) \\
& \CF(R)/\egl(R) \arrow[ur, "\scaleto{\epsilon}{4
pt}" ]
\end{tikzcd}
\end{center}
and the induced map $\epsilon$ is a surjective homomorphism of abelian groups.
\end{lemma}
\begin{proof}
    Clearly, all the maps on the diagram are homomorphisms of monoids. Factorization of $\chi$ through the stabilization is straightforward, while factorization through the action of $\egl(R)$ is due to Corollary \ref{Cor:Vfactorization}. Let us demonstrate that $\epsilon$ is surjective. Any element of $V(R)$ is a triple $[\L;q_0,q_1]$ obtained from $GW_1^{+lib}$ by taking the quotient with respect to \eqref{def:bordrelation}. Let us notice that
    \begin{align*}
        [\L;q_0,q_1]=[\L;q_0,q_1]+[\L^*;-q_0^{-1},-q_0^{-1}]=[\L\oplus \L^*;q_0\oplus -q_0^{-1},q_1\oplus -q_0^{-1}]\,.
    \end{align*}
    Homotopy \eqref{Eq:TrivMas} evaluated at $T=1$ is an elementary transformation mapping $q_0\oplus -q_0^{-1}$ to $\l^{+}$. Using this observation and Corollary \ref{Cor:Vfactorization}, we conclude that $[\L;q_0,q_1]=[\L\oplus \L^*;\l^{+},q_1\oplus -q_0^{-1}]$ as elements of $V(R)$, what demonstrates that $\epsilon$ is surjective.  
\end{proof}

\begin{theorem}\label{Thm:VandF/EGL}
    For a commutative ring with involution $R$ where $2$ is invertible, there is a homomorphism of groups:
    \begin{align}
        \rho: V(R)\longrightarrow \CF(R)/\egl(R)\
    \end{align}
    induced by the map $[\L;q_0,q_1]\mapsto  q_1\oplus -q_0^{-1}$, such that $\rho$ is inverse to $\epsilon$ and $V(R)\cong \CF(R)/\egl(R)$.
\end{theorem}
\begin{proof}
It is straightforward to check that $\rho$ is a homomorphism of abelian groups: the image of $[\L;q,q]$ in $\CF(R)$ is trivial thanks to the proof of Lemma \ref{lemma:CF/EGL}. The composition $\epsilon\circ \rho$ is identity thanks to the proof of Lemma \ref{Lem:homFtoV}. Homotopy \eqref{def:homLptoLm} evaluated at $T=1$ gives an elementary transformation, denote it by $e_{-}$, such that $e_{-}^{*}\sc \l^{+}\sc e_{-}=-\l^{+}$. Therefore, $ \rho \circ \epsilon$ is identity as well.
\end{proof}
\begin{theorem}\label{Thm:F/EGLandpiF}
    For a commutative ring with involution $R$ where $2$ is invertible there is an isomorphism of groups:
\begin{align*}
    \CF(R)/\mbox{EGL}(R)\cong \pi_0\CF(R).
\end{align*}
\end{theorem}
\begin{proof}
    If there exists $e\in \egl(\L\oplus \L^*;R)$ such that $e^*\sc \phi \sc e=\psi\in \FL(R)$, then $\phi \sim_h \psi$: a concrete homotopy can be built from any presentation of $e$ as a product of elementary matrices, analogously to the one used in the proof of Proposition \ref{Prop:EsptoPath}. This demonstrates that the map $\CF(R)/\mbox{EGL}(R)\to \pi_0\CF(R)$ is surjective. Let us demonstrate that this map is also injective. In order to do that, we shall prove that if $\phi \sim_h \psi$, as elements of $\CF(R)$, then there exists $a\in \egl(R)$ such that $a^{*}\sc \phi\sc a=\psi$. The proof below is adapted from \cite{Ojanguren} and \cite{swan1968algebraic}.

Let $\alpha(T)\in \FL(R[T])$ be such that $\alpha(0)=\phi$ and $\alpha(1)=\psi$. First of all, let us show that there exists a free module $\L'$ and $e\in \egl(\L\!\oplus\! \L^*\!\oplus \! \L'\! \oplus \! \L'^*;R[T])$ such that
\begin{align}\label{eq:linearization}
    e^{*}\sc( \alpha(T)\oplus \l^{+}_{\scaleto{\L'}{6pt}}) \sc e =\alpha_0+T\alpha_1
\end{align} 
for some $\alpha_0,\alpha_1\in \CF_{\scaleto{\L\oplus \L'}{6pt}}(R)$.  Since $2$ is invertible in $R$, as well as in $R[T]$, we can represent any $+$hermitian form as $\alpha(T)=\gamma(T)+\gamma^*(T)$ for some non-degenerate quadratic form $\gamma$. Let us assume that the degree of $\gamma$ in $T$ is $m$ and decompose $\gamma$ into a sum of homogeneous components
\begin{align*}
    \gamma=\sum_{k=0}^{m}\gamma_k T^{k}\,.
\end{align*}
Let us fix a basis for $\L\!\oplus\! \L^*$, and let $\{\gamma_{ij}\}_{i,j=1}^{r}$ be the matrix of $\gamma_m$ in this basis, then introduce a module map : 
\begin{align*}
    \kappa = 
    \begin{pmatrix}
        -T^{m-1}     &   0           &   0             & \dots & 0\\
        \gamma_{11}T & \gamma_{12} T & \gamma_{13} T   & \dots & \gamma_{1r}T\\
        0            & -T^{m-1}      &  0              & \dots & 0\\
                     &    \dots       &                 &   \dots    &  \\
        \gamma_{r1}T &  \gamma_{r2}T & \gamma_{r3}T    & \dots  & \gamma_{rr}T
    \end{pmatrix}\,.
\end{align*}
This map satisfies $\kappa^{*}\l^{+}_{\scaleto{\L'}{6pt}}\kappa=-(\gamma_m+\gamma_m^{*})T^m$ and we can use it to cancel the $T^m$ term in the homogeneous decomposition of $\alpha(T)$:
\begin{align}\label{eq:degreelow}
    \begin{pmatrix}
        1 & \kappa^{*}\\
        0 & 1 
    \end{pmatrix}
    \begin{pmatrix}
        \alpha(T) & 0\\
        0 & \l^{+} 
    \end{pmatrix}
    \begin{pmatrix}
        1 & 0\\
        \kappa & 1 
    \end{pmatrix}=
        \begin{pmatrix}
        \alpha(T)+\kappa^{*}\l^{+}\kappa & \kappa ^*\l^{+}\\
        \l^{+}\kappa & \l^{+} 
    \end{pmatrix}
\end{align}
The R.H.S. of \eqref{eq:degreelow} is an element of $\CF_{\scaleto{\L\oplus \L'}{6pt}}(R)$ with the degree at $T$ at most $m-1$.  By iterating this process, we end up with a transformation of the form \eqref{eq:linearization}, which has the degree at most one in $T$.

Let $p(T)$ be a formal series satisfying $p(T)^2=1+T$, i.e., 
\begin{align*}
    p(T)=\sum_{k=0}^{\infty}\begin{pmatrix}
    1/2\\
    k
\end{pmatrix}T^{k}\,,
\end{align*}
and let us introduce an automorphism $\xi=\alpha_0^{-1}\sc \alpha_1$, which is nilpotent, see Section 16 of \cite{swan1968algebraic}. Therefore, $p(\xi T)$ is a polynomial at $\xi T$. We notice that $\alpha_0 \sc\xi = \xi^*\sc \alpha_0$ and $\alpha_1\sc \xi = \xi^* \sc\alpha_1$ such that 
\begin{align}\label{eq:stabF}
    \left(p(\xi T)^{-1}\right)^*\sc (\alpha_0+\alpha_1T)\sc  p(\xi T)^{-1}=(\alpha_0+\alpha_1T)\sc p(\xi T)^{-2}=\alpha_0\,.
\end{align}
Further, we use the homotopy invariance of $\K_1$-functor for regular rings. Namely, if we have a path of automorphisms $a(T)\in \gl(\L;R[T])$ for some free module $\L$ such that $a(0)=1$, then $a(1)$ is elementary up to stabilization. We apply this result to $p(\xi T)^{-1}$ to demonstrate that $p(\xi)$ is elementary up to stabilization. 

Let us denote the image of $\FL(R)\ni \phi$ in $\CF(R)$ by $\st(\phi)$. Evaluation of \eqref{eq:stabF} at $1$ provides us with an elementary transformation $a\in \egl(R)$ such that $a^*\sc \st(\phi)\sc a=\st(\psi)$ which finishes the proof.
\end{proof}

As a corollary from Theorems \ref{Thm:main}, \ref{Thm:F/EGLandpiF}, and \ref{Thm:VandF/EGL}, we get an isomorphism of abelian monoids $\pi_0\Omega \CL(R)\cong V(R)$. The group $V(R)$ is a well-studied object, see \cite{KaroubiPeriodicity,KaroubiLocalization1}, while $V$ is a functor ``intermediate" between the algebraic and hermitian $\K$-theory functors \cite{Karoubi1980}. Later we shall be interested in Laurent extensions of $\IF_p$ with $p$ odd. Using techniques from \cite{KaroubiLocalization1}, one checks that $V(\IF_p)\cong \IF_p^{\times}$ and $V(\IF_p[x,x^{-1}])\cong \IZ_p^{\times}\oplus GW^{+lib}(\IF_p)$ where $GW^{+lib}(\IF_p)\cong GW^{+}(\IF_p)$; further Laurent extensions are less tractable.  

\begin{proposition}
    The following sequence is exact:
    \begin{align}\label{def:Ilib}
	K_1(R) \xlongrightarrow{\nu} V(R) \xlongrightarrow{\delta} I^{lib}(R) \longrightarrow 0
\end{align}
\end{proposition}
\begin{proof}
The kernel of $\delta$ consists of classes of triples $[\L;q_0,q_1]$ such that $[\L;q_0]-[\L;q_1]=0\in I^{lib}(R)$ what means that $q_0$ and $q_1$ are isometric, i.e., $q_1=a^*\sc q_0\sc a$ for some $a\in \gl(\L;R)$. Due to Corollary \ref{Cor:Vfactorization}, the action of $\gl(\L;R)$ factors through the homomorphism $\gl(\L;R)\to K_1(R)$.
\end{proof}

The main result of this section is the classification of loops of Lagrangians based at the standard Lagrangian $\L\subset \hm(\L)$ by pairs of non-degenerate $+$hermitian forms on $\L$. Recall that the hyperbolic $\lm$-unitaries stabilize the standard Lagrangian.  As with the group $\umod^{-}(R)$, defined by \eqref{def:UnitariesModElemModHyperb}, we treat the hyperbolic $\lm$-unitaries as trivial. In terms of $V(R)$, this amounts to taking the quotient with respect to the action of $\K_1(R)$, defined in Lemma \ref{lemma:K_1actsonV}, such that the quotient group is the fundamental ideal $I^{lib}$ of the free Grothendieck-Witt group.  In the next section, we shall review the calculation of $\switt^{+lib}(R)$ for multiple Laurent extensions of $\IF_p$ and calculate the corresponding fundamental ideals.

\section{Fundamental ideal}\label{sec:fundideal}

In this section, we analyze the fundamental ideal of the free Witt group of $+$hermitian forms for the ring of Laurent polynomials $\ifpd\coloneqq \IF_p[x_1,x_1^{-1},\dots, x_{\dd},x_{\scaleto{\dd}{5pt}}^{-1}]$ with odd $p$. Note that $I^{lib}(\ifpd)$ is isomorphic to $I(\ifpd)$ due to the Serre-Suslin-Swan theorem. Thus, it suffices to compute $I(\ifpd)$. 

Recall, the Witt group of $+$hermitian forms $\switt^+(\ifpd)$, defined in \eqref{def:Wittgroup}, is the quotient of the Grothedieck-Witt group of $+$hermitian forms on finite-dimensional free $\ifpd$-modules modulo hyperbolic forms. The fundamental ideal $I(\ifpd)$ consists of elements from $\switt^+(\ifpd)$ corresponding to classes of even-dimensional forms. Thus, $I$ fits into the following exact sequence 
\begin{align}
	0 \longrightarrow I(\ifpd) \longrightarrow \switt^+(\ifpd)\longrightarrow \IZ/2 \longrightarrow 0
\end{align}
such that $I(\ifpd)$ is a subgroup of $\switt^+(\ifpd)$ of order two.

 It is instructive to calculate $I(\IF_p)$ explicitly. It is well-known, see, e.g., \cite{milnor2013symmetric}, that the Witt group $\switt^+(\IF_p)$ is a group of order four generated by +hermitian forms on $1$-dimensional modules. Let us denote those forms by $\la 1\ra$ and $\la \theta \ra$, where $\theta \in \IF_p^{\times}/\IF_p^{2\times}$ is a square non-residue. Then, the Witt group is given by the following generators and relations depending on $p$ mod $4$:
 \begin{align}
 \switt^{+}(\IF_p)&\cong \langle \;\la 1\ra ,\la \theta\ra\,|\,2\la 1\ra=2\la \theta\ra=0 \;\rangle \cong \IZ/2\oplus \IZ/2\,,\quad\quad\quad\quad\quad\,\;\; p\equiv 1\;\mbox{mod}\;4\,,\; \nonumber\\
\switt^{+}(\IF_p)&\cong \langle \;\la 1\ra ,\la \theta\ra\,|\,2\la 1\ra=2\la \theta\ra, \la1\ra +\la \theta \ra =0 \;\rangle  \cong \IZ/4\,,\quad\quad\quad\;\;\;\;\; p\equiv 3\;\mbox{mod}\;4\,.\;  \nonumber    
 \end{align}
Thus, $I(\IF_p)\cong \IZ/2$ is generated by $\la 1\ra +\la \theta\ra$ for $p\equiv 1$ mod $4$ and by $2\la \theta\ra  $ for $p\equiv 3$ mod $4$. This can also be seen directly from \eqref{def:Ilib} after noting that $V(\IF_p)\cong \IF^{\times}_p$.

We are interested in the fundamental ideal $I(\IF_p^{\dd})$ for $\dd=1,2,3$, and $4$. The calculation of $\switt^{+}(\ifpd)$ was done by A. Ranicki in \cite{Ra2} and recently reviewed by J. Haah in \cite{Haah2022}. Due to the Serre-Suslin-Swan theorem, Ranicki's $V$- and $U$-theories for $\ifpd$ coincide. Also, theories of $\pm$hermitian forms coincide with theories of $\pm$quadratic forms since we focus on $\ifpd$ with odd $p$. We adopt the notation used by Haah: 
\begin{align*}
    \vtheory_n(\IF_p^{\dd})=\begin{cases}
        \switt^{+}(\ifpd)\,,\quad n\equiv 0\;\mbox{mod}\;4\,,\\
        \umod^{+}(\ifpd)\,,\quad \,\,n\equiv 1\;\mbox{mod}\;4\,,\\
        \switt^{-}(\ifpd)\,,\quad n\equiv 2\;\mbox{mod}\;4\,,\\
        \umod^{-}(\ifpd)\,,\quad\,\, n\equiv 3\;\mbox{mod}\;4\,.
        \end{cases}
\end{align*}
These L-groups satisfy the following recursion relation \cite{Ra2, Haah2022}:
\begin{align}\label{eq:descent}
    \vtheory_n(\IF_p^{\dd})\cong \vtheory_n(\IF_p^{\dd-1})\oplus \vtheory_{n-1}(\IF_p^{\dd-1})\,.
\end{align}
The base of recursion is given by $\vtheory_0(\IF_p)=\switt^{+}(\IF_p)$ and by known isomorphisms:
\begin{align*}
    \vtheory_1(\IF_p)\cong 0\,\quad  \vtheory_2(\IF_p)\cong 0\,\quad\vtheory_3(\IF_p)\cong 0\,.
\end{align*}
We immediately obtain:
\begin{align}\label{eq:4dWittgroup}
\switt^+(\IF_p)\cong \switt^+(\IF_p^1)\cong \switt^+(\IF_p^2)\cong \switt^+(\IF_p^3)\,,\quad \switt^+(\IF_p^4)\cong \switt^+(\IF_p)\oplus \switt^+(\IF_p)\,,
\end{align}
and obvious isomorphisms
 \begin{align}\label{eq:0123fundideals}
I(\IF_p)\cong I(\IF_p^1)\cong I(\IF_p^2)\cong I(\IF_p^3)\,.
\end{align}
In order to calculate $I(\IF_p^4)$, we need a closer look at how  isomorphism \eqref{eq:descent} works. We always have a ring homomorphism of evaluation at $1$, which induces a group homomorphism $\switt^+(\IF_p^4)\to \switt^+(\IF_p^3)$ whose kernel is isomorphic to $\switt^{+}(\IF_p)$. Moreover, one copy of $\switt^{+}(\IF_p)$ is present in $\switt^{+}(\IF_p^\dd)$ for all $\dd$ and it comes from $\dd=0$. The second copy of $\switt^{+}(\IF_p)$ is the image of $\umod^{-}(\IF_p^3)$ under a map $\bass^\uparrow_{3}$:
\begin{align}
    \bass^\uparrow_{3} : \umod^{-}(\IF_p^3) \to  \switt^{+}(\IF_p^4)
\end{align}
defined in Section 5.8 of \cite{Haah2022}. In short, $\bass^\uparrow_{3}$ maps a $\lambda^{-}$-unitary $u\in \hp^{-}(\L;R)$ to a $+$hermitian form on the induced $\IF_p^{4}$-module $\IF_{p}^{4}\otimes_{\IF_{p}^{3}}(\L\oplus \L^*)$. Explicitly, the non-trivial classes of forms in $\switt^{+}(\IF_p^4)$ can be obtained by applying $\bass^\uparrow_{3}$ to the non-trivial classes of $\umod^{-}(\IF_p^3)$ which are provided in Section IV of \cite{Haah2021}. This simple observation leads to a conclusion that the whole image of $\bass^\uparrow_{3}$ consists of classes of even-dimensional hermitian forms and we obtain
\begin{align}\label{4dfundideal}
    I(\IF_p^4)\cong \IZ/2\oplus \switt^{+}(\IF_p)\,.
\end{align}

Note that the first term $\IZ/2$ is also present in the zero-dimensional case, i.e., there is an inclusion $\iota:I(\IF_p)\to I(\IF_p^\dd)$ for $\dd=1,2,3$. Clearly, any loop of Lagrangians over $\IF_p$ defines a loop of Lagrangians over $\ifpd$. In order to get rid of the zero-dimensional loops, we define $ \tilde I(\IF_p^\dd)=\mbox{coker}\,\iota$. A similar reduction of the space of invertible lattice systems is discussed in Section VI of \cite{wen2022flow}. As a result, we obtain that the homotopy classes of loops of free Lagrangians modulo shifts and zero-dimensional loops are given by 
\begin{align}
    \tilde I(\IF_p^\dd)\cong 0\;\;\mbox{for}\;\; \dd=0,1,2,3\,,\quad \mbox{and}\mbox\quad  \tilde I(\IF_p^4)\cong \switt^+(\IF_p)\,.
\end{align}
This concludes our calculation and we arrive at the result announced in Section \ref{intro}.

    \section{Discussion}\label{sec:discussion}
In this paper, we studied homotopy classes of periodic Clifford unitaries through the analysis of stabilizer Hamiltonians they act on. We found that there are non-trivial homotopy classes of periodic Clifford unitary dynamics in four spatial dimensions.  

\textbf{1.} Our calculations heavily relied on the classical results in Hermitian K-theory, in particular, the fundamental theorem by Karoubi \cite{Karoubi1980} and its interpretation by Barges and Lannes \cite{BL}. The fundamental theorem applies only to rings with involution containing an element $\lambda$ such that $\l+\bar \l=2$. In our case, this condition restricts us to prime $p$-dimensional qudits with $p>2$. However, from the physics perspective, we expect that similar results can be obtained for $p=2$. Namely, we expect that the group of homotopy classes of reduced loops of Clifford unitaries for four-dimensional lattices of qubits is isomorphic to the Witt group of symmetric bilinear forms on $\IF_2$-vector spaces. The reason to believe it is true is described in the next paragraph.

\textbf{2.} The results obtained in this letter confirm the general correspondence between $\dd$-dimensional QCA and $(\dd+1)$-dimensional Floquet circuits, also known as the bulk-boundary correspondence, as mentioned in \cite{Haah2022} and analyzed for $\dd=1$ in  \cite{zhang2022bulkboundary}. A heuristic argument towards the bulk-boundary correspondence for QCA is a variant of the famous argument by Kitaev about loops of short-range entangled phases \cite{Kitaev}, reviewed, in particular, in \cite{Shiozaki2022}. Colloquially speaking, we should replace the $\dd$-dimension SRE state in Kitaev's argument with a $\dd$-dimensional QCA to obtain a $(\dd+1)$-dimensional Floquet circuit parametrized by the QCA. In our approach, the bulk-boundary correspondence arises as an incarnation of the fundamental theorem of hermitian K-theory.

\textbf{3.} As we reviewed in the introductory section, the framework of Clifford unitaries admits a formulation in terms of $\lm$-unitaries and is naturally connected to the Hermitian K-theory. Of course, there are QCA beyond Clifford QCA, see examples in \cite{Wilbur2022}, for which the Hermitian K-theory does not immediately apply. A novel approach to the general QCA in terms of invertible subalgebras can be found in \cite{haah2022invertible} and $C^*$-algebraic approach to symmetric QCA is discussed in \cite{jones2023dhr}. One of the possible candidates describing the structure of general QCA is a Hermitian K-theory of dg categories in the sense of \cite{schlichting2017hermitian}.

\bibliographystyle{unsrt}
\bibliography{CliffLoopsAlt}

\begin{thebibliography}{10}

\bibitem{Nahum2017}
Adam Nahum, Jonathan Ruhman, Sagar Vijay, and Jeongwan Haah.
\newblock Quantum entanglement growth under random unitary dynamics.
\newblock {\em Physical Review X}, 7(3):031016, 2017.

\bibitem{Nahum2018}
Adam Nahum, Sagar Vijay, and Jeongwan Haah.
\newblock Operator spreading in random unitary circuits.
\newblock {\em Physical Review X}, 8(2):021014, 2018.

\bibitem{Li2019}
Yaodong Li, Xiao Chen, and Matthew~PA Fisher.
\newblock Measurement-driven entanglement transition in hybrid quantum
  circuits.
\newblock {\em Physical Review B}, 100(13):134306, 2019.

\bibitem{Fisher2023}
Matthew~PA Fisher, Vedika Khemani, Adam Nahum, and Sagar Vijay.
\newblock Random quantum circuits.
\newblock {\em Annual Review of Condensed Matter Physics}, 14:335--379, 2023.

\bibitem{Mi2021}
Xiao Mi, Pedram Roushan, Chris Quintana, Salvatore Mandra, Jeffrey Marshall,
  Charles Neill, Frank Arute, Kunal Arya, Juan Atalaya, Ryan Babbush, et~al.
\newblock Information scrambling in quantum circuits.
\newblock {\em Science}, 374(6574):1479--1483, 2021.

\bibitem{farshi2022mixing}
Tom Farshi, Daniele Toniolo, Carlos~E Gonz{\'a}lez-Guill{\'e}n, {\'A}lvaro~M
  Alhambra, and Lluis Masanes.
\newblock Mixing and localization in random time-periodic quantum circuits of
  clifford unitaries.
\newblock {\em Journal of Mathematical Physics}, 63(3), 2022.

\bibitem{farshi2023absence}
Tom Farshi, Jonas Richter, Daniele Toniolo, Arijeet Pal, and Lluis Masanes.
\newblock Absence of localization in two-dimensional clifford circuits.
\newblock {\em PRX Quantum}, 4(3):030302, 2023.

\bibitem{schlingemann2008structure}
Dirk-M Schlingemann, Holger Vogts, and Reinhard~F Werner.
\newblock On the structure of clifford quantum cellular automata.
\newblock {\em Journal of Mathematical Physics}, 49(11), 2008.

\bibitem{Freedman2020}
Michael Freedman and Matthew~B Hastings.
\newblock Classification of quantum cellular automata.
\newblock {\em Communications in Mathematical Physics}, 376:1171--1222, 2020.

\bibitem{Hastings2013}
Matthew~B Hastings.
\newblock Classifying quantum phases with the kirby torus trick.
\newblock {\em Physical Review B}, 88(16):165114, 2013.

\bibitem{Lukasz2020}
Lukasz Fidkowski, Jeongwan Haah, and Matthew~B Hastings.
\newblock Exactly solvable model for a 4+ 1 d beyond-cohomology
  symmetry-protected topological phase.
\newblock {\em Physical Review B}, 101(15):155124, 2020.

\bibitem{Haah2022}
Jeongwan Haah.
\newblock Topological phases of unitary dynamics: Classification in clifford
  category.
\newblock {\em arXiv preprint arXiv:2205.09141}, 2022.

\bibitem{ogata2021classification}
Yoshiko Ogata.
\newblock Classification of gapped ground state phases in quantum spin systems.
\newblock {\em arXiv preprint arXiv:2110.04675}, 2021.

\bibitem{kapustin2022local}
Anton Kapustin and Nikita Sopenko.
\newblock Local noether theorem for quantum lattice systems and topological
  invariants of gapped states.
\newblock {\em Journal of Mathematical Physics}, 63(9), 2022.

\bibitem{beaudry2023homotopical}
Agnes Beaudry, Michael Hermele, Juan Moreno, Markus Pflaum, Marvin Qi, and
  Daniel Spiegel.
\newblock Homotopical foundations of parametrized quantum spin systems.
\newblock {\em arXiv preprint arXiv:2303.07431}, 2023.

\bibitem{chung2023topological}
Jui-Hui Chung and Jacob Shapiro.
\newblock Topological classification of insulators: I. non-interacting
  spectrally-gapped one-dimensional systems.
\newblock {\em arXiv preprint arXiv:2306.00268}, 2023.

\bibitem{haah2013}
Jeongwan Haah.
\newblock Commuting pauli hamiltonians as maps between free modules.
\newblock {\em Communications in Mathematical Physics}, 324:351--399, 2013.

\bibitem{Karoubi1980}
Max Karoubi.
\newblock Le th{\'e}oreme fondamental de la k-th{\'e}orie hermitienne.
\newblock {\em Annals of mathematics}, 112(2):259--282, 1980.

\bibitem{BL}
Jean Barge and Jean Lannes.
\newblock {\em Suites de Sturm, indice de Maslov et périodicité de Bott}.
\newblock Progress in Mathematics. Birkhäuser Basel, 2008.

\bibitem{kitaev2002classical}
Alexei~Yu Kitaev, Alexander Shen, and Mikhail~N Vyalyi.
\newblock {\em Classical and quantum computation}.
\newblock Number~47. American Mathematical Soc., 2002.

\bibitem{Havlicek2002}
Miloslav Havlicek, Jiri Patera, Edita Pelantova, and Jiri Tolar.
\newblock Automorphisms of the fine grading of sl (n, c) associated with the
  generalized pauli matrices.
\newblock {\em Journal of Mathematical Physics}, 43(2):1083--1094, 2002.

\bibitem{haah2016}
Jeongwan Haah.
\newblock Algebraic methods for quantum codes on lattices.
\newblock {\em Revista colombiana de matematicas}, 50(2):299--349, 2016.

\bibitem{KBook}
Charles~A Weibel.
\newblock {\em The $ K $-book: An Introduction to Algebraic $ K $-theory},
  volume 145.
\newblock American Mathematical Soc., 2013.

\bibitem{KaroubiPeriodicity}
Max Karoubi.
\newblock P{\'e}riodicit{\'e} de la {K}-{T}h{\'e}orie hermitienne.
\newblock {\em Lecture Notes in Math}, 237(343):301--411, 1973.

\bibitem{Suslin1977}
Andrei~A Suslin.
\newblock On the structure of the special linear group over polynomial rings.
\newblock {\em Mathematics of the USSR-Izvestiya}, 11(2):221, 1977.

\bibitem{Swan1978}
Richard~G Swan.
\newblock Projective modules over laurent polynomial rings.
\newblock {\em Transactions of the American Mathematical Society},
  237:111--120, 1978.

\bibitem{voevodski}
Vladimir Voevodsky.
\newblock A1-homotopy theory.
\newblock In {\em Proceedings of the international congress of mathematicians},
  volume~1, pages 579--604. Berlin, 1998.

\bibitem{haah2023nontrivial}
Jeongwan Haah, Lukasz Fidkowski, and Matthew~B Hastings.
\newblock Nontrivial quantum cellular automata in higher dimensions.
\newblock {\em Communications in Mathematical Physics}, 398(1):469--540, 2023.

\bibitem{ghys2015signatures}
{\'E}tienne Ghys and Andrew Ranicki.
\newblock Signatures in algebra, topology and dynamics.
\newblock {\em arXiv preprint arXiv:1512.09258}, 2015.

\bibitem{Ojanguren}
Manuel Ojanguren.
\newblock On karoubi's theorem: $ w (a)= w (a [t]) $.
\newblock {\em Archiv der Mathematik}, 43(ARTICLE):328--331, 1984.

\bibitem{swan1968algebraic}
Richard~G. Swan.
\newblock {\em Algebraic {K}-{T}heory}.
\newblock Lecture Notes in Mathematics. Springer Berlin Heidelberg, 1968.

\bibitem{KaroubiLocalization1}
Max Karoubi.
\newblock Localisation de formes quadratiques {I}.
\newblock {\em Annales scientifiques de l'\'Ecole Normale Sup\'erieure},
  7(3):359--403, 1974.

\bibitem{milnor2013symmetric}
John Milnor and Dale Husem{\"o}ller.
\newblock {\em Symmetric {B}ilinear {F}orms}.
\newblock Ergebnisse der Mathematik und ihrer Grenzgebiete. 2. Folge. Springer
  Berlin Heidelberg, 2013.

\bibitem{Ra2}
Andrew~A. Ranicki.
\newblock Algebraic {L}-{T}heory, {II}: {L}aurent extensions.
\newblock {\em Proceedings of the London Mathematical Society}, 3(1):126--158,
  1973.

\bibitem{Haah2021}
Jeongwan Haah.
\newblock Clifford quantum cellular automata: Trivial group in 2d and witt
  group in 3d.
\newblock {\em Journal of Mathematical Physics}, 62(9), 2021.

\bibitem{wen2022flow}
Xueda Wen, Marvin Qi, Agnes Beaudry, Juan Moreno, Markus~J Pflaum, Daniel
  Spiegel, Ashvin Vishwanath, and Michael Hermele.
\newblock Flow of higher berry curvature and bulk-boundary correspondence in
  parametrized quantum systems.
\newblock {\em Physical Review B}, 108(12):125147, 2023.

\bibitem{zhang2022bulkboundary}
Carolyn Zhang and Michael Levin.
\newblock Bulk-boundary correspondence for interacting floquet systems in two
  dimensions.
\newblock {\em Physical Review X}, 13(3):031038, 2023.

\bibitem{Kitaev}
Alexei Kitaev.
\newblock On the classification of short-range entangled states.
\newblock {\em talk at Simons Center for Geometry and Physics}, 2013.

\bibitem{Shiozaki2022}
Ken Shiozaki.
\newblock Adiabatic cycles of quantum spin systems.
\newblock {\em Physical Review B}, 106(12):125108, 2022.

\bibitem{Wilbur2022}
Wilbur Shirley, Yu-An Chen, Arpit Dua, Tyler~D Ellison, Nathanan
  Tantivasadakarn, and Dominic~J Williamson.
\newblock Three-dimensional quantum cellular automata from chiral semion
  surface topological order and beyond.
\newblock {\em PRX Quantum}, 3(3):030326, 2022.

\bibitem{haah2022invertible}
Jeongwan Haah.
\newblock Invertible subalgebras.
\newblock {\em Communications in Mathematical Physics}, 403(2):661--698, 2023.

\bibitem{jones2023dhr}
Corey Jones.
\newblock Dhr bimodules of quasi-local algebras and symmetric quantum cellular
  automata.
\newblock {\em arXiv preprint arXiv:2304.00068}, 2023.

\bibitem{schlichting2017hermitian}
Marco Schlichting.
\newblock Hermitian k-theory, derived equivalences and karoubi's fundamental
  theorem.
\newblock {\em Journal of Pure and Applied Algebra}, 221(7):1729--1844, 2017.

\end{thebibliography}
\end{document}